\documentclass[journal]{IEEEtran}
\usepackage{amsmath,amsfonts}
\usepackage{algorithmic}
\usepackage{array}
\usepackage[caption=false,font=footnotesize,labelfont=rm,textfont=rm]{subfig}

\usepackage{textcomp}
\usepackage{stfloats}
\usepackage{url}
\usepackage{verbatim}
\usepackage{graphicx}

\usepackage{booktabs}
\usepackage{algorithm}
\usepackage{amssymb}
\usepackage{siunitx}
\usepackage{authblk}
\usepackage{amsthm}
\newtheorem{theorem}{Theorem} 

\newtheorem{lemma}{Lemma}
\newtheorem{definition}{Definition}%
\usepackage{algorithmic}\makeatletter
\newcommand{\removelatexerror}{\let\@latex@error\@gobble}
\makeatother

\def\BibTeX{{\rm B\kern-.05em{\sc i\kern-.025em b}\kern-.08em
    T\kern-.1667em\lower.7ex\hbox{E}\kern-.125emX}}
\usepackage{balance}
\begin{document}
\title{An Adaptive Hybrid Channel Reservation Medium Access Control Protocol for Differentiated QoS}
\author{Ze Liu, Bo Li, Mao Yang\textsuperscript{*}, Zhongjiang Yan, Xichan Liu\thanks{Ze Liu, Bo Li, Mao Yang, Zhongjiang Yan and Xichan Liu are with School of Electronic Information, the Northwestern Polytechnical University, Xian, China. Email: {yangmao}@mail.nwpu.edu.cn.}}



\markboth{}%
{How to Use the IEEEtran \LaTeX \ Templates}

\maketitle
\begin{abstract}
In a densely deployed distributed wireless network, there may be various types of traffic with differentiated Quality of Service (QoS) requirements. However, when the network is heavily loaded, the collision increases significantly, making it difficult to guarantee the QoS of traffic. Designing an efficient Medium Access Control (MAC) protocol to guarantee the QoS of different types of traffic is an essential research direction. Channel reservation mechanism is a promising approach to improving QoS.
 However, few studies have focused on the channel reservation mechanism for differentiated traffic. It is difficult to take into account both the QoS of real-time traffic and the collision issue for ordinary traffic. To address this issue, this paper proposes the Differentiated Service Guarantee Adaptive Reservation Mechanism (DSGARM) protocol. A hybrid reservation mechanism is proposed by combining the absolute reservation mechanism and the relative reservation mechanism. The absolute reservation mechanism is adopted for real-time traffic. Meanwhile, the relative reservation mechanism is adopted for ordinary traffic. An adaptive algorithm is proposed to calculate the reservation parameters that meet the delay requirements based on the network conditions. The proposed work can be widely applied in the densely deployed distributed wireless network with differentiated QoS requirements. In addition, this paper establishes a mathematical model for the proposed mechanism and theoretically analyzes the performance.
Simulations verify that the mathematical model provides a good approximation of the protocol performance and demonstrates the advantages of the proposed protocol. 
\end{abstract}

\begin{IEEEkeywords}
Distributed wireless network, channel reservation mechanism, differentiated traffic guarantee.
\end{IEEEkeywords}

\section{Introduction}

\IEEEPARstart{W}{ith} the rapid development of the wireless communication technology, various densely deployed distributed wireless network scenarios have emerged. These scenarios, such as the Internet of Things, high-density Wireless Local Area Networks (WLAN), and mesh networks, are characterized by the deployment of a large number of devices in a limited area, leading to a high degree of interference and congestion \cite{barrachina2019komondor}. Real-time applications, such as voice-over-internet protocol, online games, and virtual reality (VR), pose significant challenges for wireless networks. The QoS of this traffic can be fully  guaranteed only when the delay satisfies specific conditions. In such scenarios, the transmission of traffic from multiple nodes entails distinct QoS requirements \cite{intro2}. For instance, traffic such as emergency messages needs to be transferred immediately with high reliability while other message such as surfing messages do not require such low delay requirement. Therefore, network protocols must provide differentiated QoS guarantee to accommodate different types of data \cite{intro3}.


IEEE 802.11ax utilizes the Enhanced Distributed Channel Access (EDCA) mechanism to differentiate different types of traffic and satisfy their respective QoS requirements. This mechanism employs varying contention window sizes for different types of traffic. However, in densely deployed distributed scenarios, using a smaller contention window for real-time traffic will lead to an increase in collision probability, delay time, packet loss rate, resulting in a poor QoS in the performance of real-time traffic. As highlighted in research findings \cite{intro4,2018Wireless}, EDCA is even less efficient than DCF, the predecessor of EDCA, in larger networks. 


One significant reason for the issues outlined above is the collision caused by disordered contention. To address this, channel reservation MAC schemes that aim to reduce disordered contention by reserving channel resources have been proposed \cite{intro8,intro9}. Channel reservation enable nodes to reserve future transmission periods based on their traffic status after obtaining a transmission opportunity. Then neighbor nodes keep silent in the reserved period to reduce collision. This mechanism provides a means to address the challenges posed by traditional random access methods, particularly under high network load conditions.  


The existing reservation mechanisms can be classified into absolute and relative reservation mechanisms. The absolute reservation mechanism can accurately indicate the reservation period, but there is a fragmentation issue that leads to resource waste. Relative reservation mechanism does not have this disadvantage, but cannot guarantee that the reservation period strictly meets the QoS requirements. Neither of these reservation mechanisms can satisfy the differentiated QoS requirements of various traffic. Especially, in densely deployed scenarios where multiple traffic with different QoS requirements coexist, the existing channel reservation mechanisms struggle to guarantee differentiated QoS requirements. Most reservation mechanisms only ensure the transmission performance of real-time traffic, while collisions for ordinary traffic remain unresolved. On the other hand, the reservation parameters for real-time traffic are typically specified directly and cannot guarantee QoS strictly.

To address the aforementioned problems, this paper proposes an adaptive hybrid channel reservation MAC protocol with differentiated QoS guarantees. The proposed protocol provides differentiated guarantees for different types of traffic by using corresponding reservation mechanisms and adaptive reservation parameter determination algorithms. The contributions of this paper can be summarized as follows:


(1) To the best of our knowledge, this is the first work to systematically combine the absolute reservation mechanism and the relative reservation mechanism to provide differentiated QoS guarantees for different types of traffic. The proposed protocol can reduce the collision of ordinary traffic while guaranteeing the delay requirement of real-time traffic and the performance superiority of the proposed protocol is verified by simulation.

(2) This paper establishes mathematical models and conducts theoretical analyses on the proposed hybrid mechanism through Markov models. The model derives the simulation result with different reservation mechanisms and parameters, as well as considering the effect of the non-ideal channel.

(3) The adaptive reservation parameter determination algorithm is proposed based on the established mathematical model, which provides differentiated QoS guarantees for different types of traffic based on their tolerable delay requirements, network conditions, and priority levels.

The remainder of this article is organized as follows: Section \ref{RELATED WORKS} provides an overview of existing research related to the topic. The system model is introduced In Section \ref{System Model} and the DSGARM protocol process is presented in Section \ref{PROPOSED DSGARM PROTOCOL PROCESS}. Section \ref{PERFORMANCE ANALYSIS} presents the theoretical analysis model. In Section \ref{SIMULATION RESULT}, the comparison between the simulation and theoretical results, as well as the proposed scheme and other reservation mechanisms are described detailedly. Some implementation issues are covered in Section \ref{IMPLEMENTATION ISSUES}. Finally, Section \ref{Conclusions} summarizes the work of this paper.

\section{Related Works}\label{RELATED WORKS}

\begin{figure}[!t]
\centering
\includegraphics[width=3in]{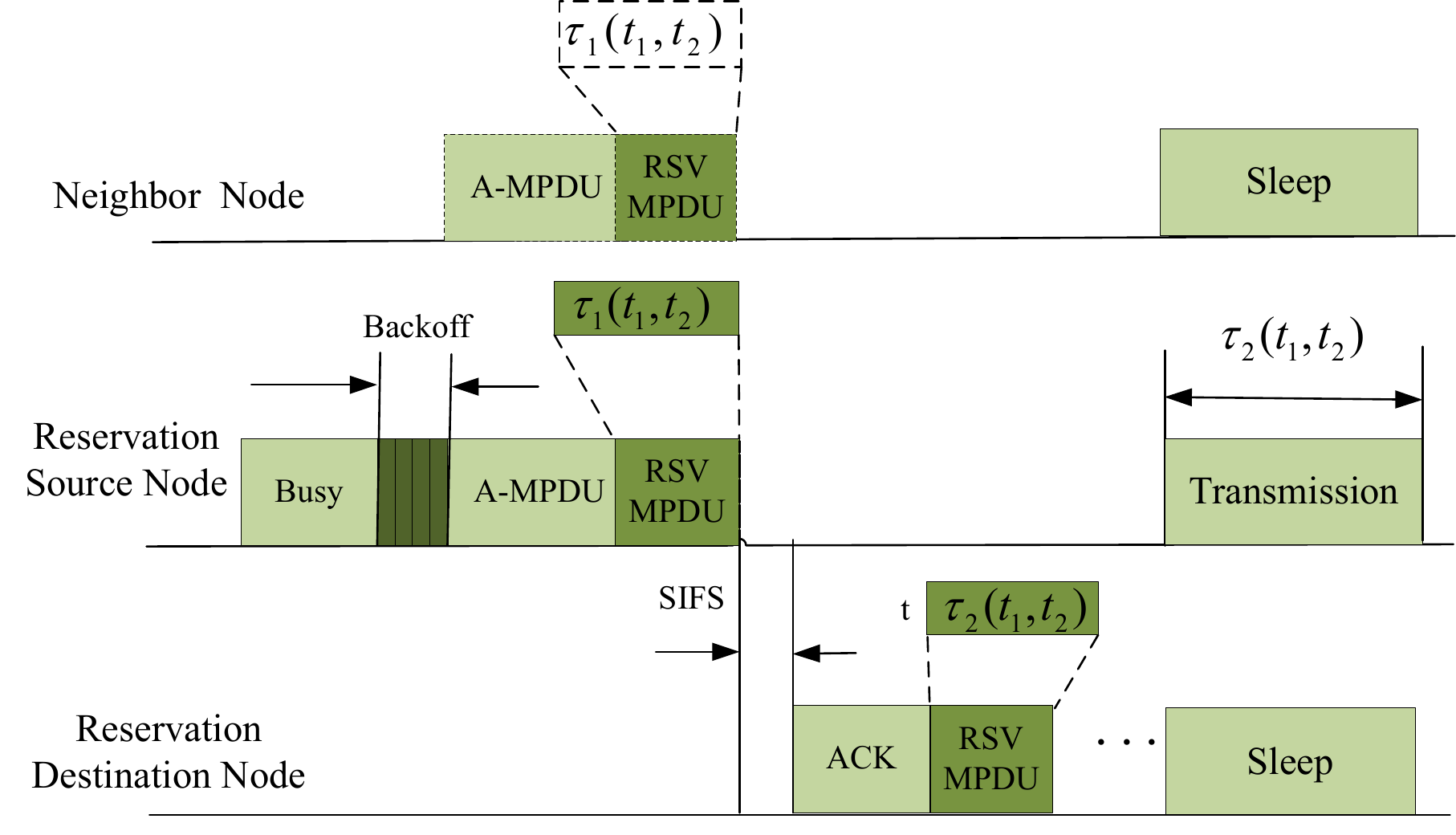}
    \caption{Basic channel reservation process}
	\label{Basic channel reservation process} 
\end{figure}

Fig. \ref{Basic channel reservation process} shows the basic process of the channel reservation mechanism. After obtaining the transmission opportunity, if there is any packet remaining in the queue, the reservation source node initiates the reservation process. The reservation MAC protocol data unit (RSV MPDU) which indicates the reservation period is then performed and will be inserted into the data frame that is currently being transmitted. The neighbor nodes which receive the reservation frame will remain silent at the reservation time. To overcome the issue of the hidden terminal, the reservation destination node needs to insert the RSV MPDU into the ACK, allowing the hidden terminal of the reservation source node to receive the reservation information.

\subsection{Absolute Reservation Mechanism}\label{Absolute Reservation Mechanism}
\begin{figure}[htbp]
\centering
\includegraphics[width=3.5in]{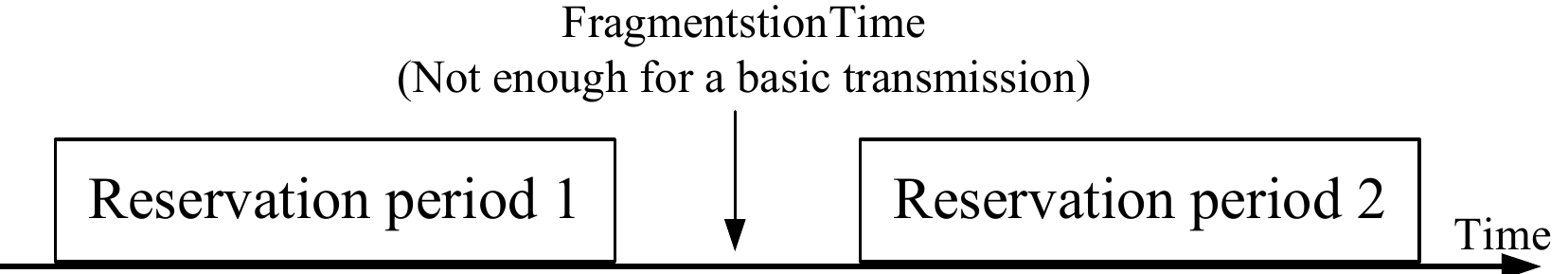}
    \caption{The fragmentation problem}
	\label{Fragmentation time description}
\end{figure}
The absolute reservation mechanism indicates the reservation period by explicitly stating the start and end time. This mechanism accurately determines the reservation offset time and efficiently guarantees QoS requirements. However, it is susceptible to fragmentation, as illustrated in Fig. \ref{Fragmentation time description}. In scenarios where the duration between two consecutive reservation periods is small, a basic transmission might be insufficient to access the channel, resulting in the wastage of channel resources. Due to the distributed approach employed by the reservation protocol in selecting reservation time, mitigating the issue of fragmentation time becomes challenging \cite{juedui1}.

Currently, most protocols using the absolute reservation mechanism simply divide traffic into two categories and make reservations for real-time traffic. Some researches \cite{juedui2,juedui3}, proposed reservation mechanisms for real-time traffic that improve channel utilization by releasing unused channel resources in a timely manner. These reservation mechanisms effectively guarantee the delay of high-priority traffic, but the collision issues of ordinary traffic are still difficult to resolve. Moreover, the fragmentation problem of reservation resources may impact the performance of ordinary traffic.

To accurately ensure the traffic quality requirements, some researches proposed more detailed reservations based on the minimum bandwidth or delay requirements of the traffic. For example, research \cite{juedui4} presents a hybrid protocol that divides the superframe into the high-priority traffic reservation periods and the low-priority traffic contention periods. The central control node allocates the resources of the reservation period according to the minimum bandwidth requirement, and the low-priority traffic accesses the channel through random contention. This mechanism assigns reservation resources based on QoS requirements and effectively ensures the real-time traffic QoS requirements when the network load is light. However, it is limited by the inability to adjust the reservation parameters according to the network conditions and the lack of considering the collision problem of ordinary traffic while ensuring the traffic requirements of real-time traffic.
Researches \cite{juedui5,juedui6} establish mathematical models to analyze the impact of reservation parameters and allocated bandwidth on the average traffic delay and successful transmission rate. The study aimed to determine the appropriate reservation parameters for optimal performance when there is a single type of real-time traffic in the network. This work can be combined with our study to conduct a comprehensive analysis of the reservation mechanism.
\subsection{Relative Reservation Mechanism}\label{Relative Reservation Mechanism}

The research \cite{xiangdui1} describes a typical relative reservation mechanism. The reservation source node indicates the residual backoff value in the transmitted packet. The neighbor node records the residual backoff value and maintains it according to CSMA. It will select the unreserved backoff value when the transmission is initiated. If all neighbor nodes correctly receive the reservation information, the neighbor nodes will not select the same backoff value as the reservation source node, to avoid collision.


The relative reservation mechanism carries out reservation through selecting the backoff value and does not divide the channel resources into multiple periods to avoid the fragmentation problems. However, it is difficult to accurately determine the reservation time using the backoff value, as the specific reservation time is related to the proportion of the busy channel duration. When other nodes frequently access the channel, the reservation time may be delayed. Therefore, for real-time traffic, the relative reservation mechanism can only reduce the collision but cannot guarantee QoS strictly.

Among the existing relative reservation mechanisms, some related protocols employ differentiated reservations for various traffic types. Researches \cite{xiangdui2,xiangdui3} are based on a deterministic Semi-Random backoff reservation mechanism. That is, after a STA successfully obtains the channel transmission opportunity, the AP assigns it a collision-free fixed backoff value in the ACK. Different traffic types use different fixed backoff values to achieve priority. However, the fixed backoff value of each traffic is still directly specified according to EDCA parameters and is not adjusted according to the actual business requirements. Although priority is introduced into the channel reservation mechanism, the differential parameters are directly specified based on priority. When the number of nodes is large and the network load is heavy, the directly specified backoff value cannot accurately guarantee the QoS requirements.


\subsection{Problem Analysis of the Existing Researches}\label{Existing reservation mechanism problems}
The existing reservation mechanisms have the following problems and have been fully considered in this paper.

(1) The existing researches used one of the absolute and the relative resevation, which can not meet the different requirements of real-time traffic and ordinary traffic.

(2) Most absolute reservation mechanisms divide traffic into two priority levels and only guarantee the high-priority traffic. Therefore, the transmission performance of ordinary traffic is often not guaranteed. Meanwhile, even in high-priority traffic, different requirements may exist. The absolute reservation mechanism has its own fragmentation problems. Therefore there is a threat to guarantee requirements of all types of traffic.

(3) Some relative reservation mechanisms differentiates traffic priorities and reserves resources accordingly. The differentiated parameters are usually directly specified based on priority and is difficult to accurately guarantee QoS requirements.


\section{System Model}\label{System Model}

Consider a densely deployed distributed wireless network, where all nodes access the channel through the proposed CSMA-based channel reservation mechanism. The protocol is primarily designed for high-load scenarios with dense deployment, where each traffic has different QoS requirements.

The proposed protocol employs a channel reservation mechanism based on the CSMA. The network under consideration has $k_{tot}$ types of traffic, where exists $k_{low}$ types of real-time traffic and $k_{ord}$ types of ordinary traffic with no strict delay requirements.


\begin{figure}[htbp]
\centering
\includegraphics[width=1.2in]{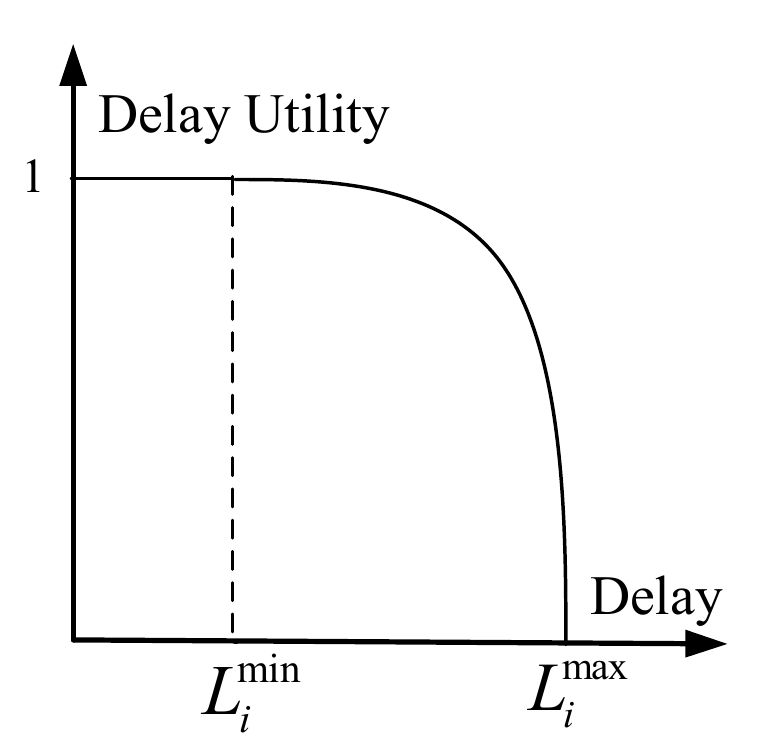}
    \caption{The relation between delay utility and delay}
	\label{Calculation of traffic delay utility}
\end{figure}

In wireless networks, utility serves as a metric to quantify the level of user satisfaction with transmission performance. Each traffic type is associated with a specific utility function that characterizes its performance. The value of the utility function approaches $1$ if the delay is less than the minimum tolerance delay ($L_i^{min}$) and gradually approaches $0$ if the delay exceeds the maximum tolerance delay ($L_i^{max}$) \cite{uti1}. The existing research \cite{uti2} has proposed the following Z-shaped function for calculating the delay utility,

\begin{equation*}
{U_i(T_i)} = \begin{cases}
1,&{T_i\in(0,L_i^{min})} \\ 
1-e^{-k(T_i-L_i^{max})^2},&T_i\in[L_i^{min},L_i^{max}] \\ 
{0,}&T_i>L_i^{max}
\end{cases}
,
\end{equation*}
where $k$ is a constant satisfying $1-e^{-k(T-L_i^{max})^2}=1$ and $T_i$ represent the delay of traffic $i$.


\section{Proposed DSGARM Protocol Process}\label{PROPOSED DSGARM PROTOCOL PROCESS}

In this section, we present the proposed hybrid protocol that combines absolute reservation and relative reservation mechanisms. 
Firstly, we give some related definitions before our further discussions.
\begin{definition}\label{definitionReservationWeigh}
Reservation Weight

Each traffic is assigned a reservation weight, denoted by $\alpha_i$, which reflects its importance. The reservation parameters are calculated based on the weight.
\end{definition}
\begin{definition}\label{definition5}
Reservation blocking status

Assuming the maximum interval between two consecutive reservation transmissions is $k$ periods, which is determined by the maximum delay requirements of the traffic, we define $k$ as the reservation blocking quantity. If all the first $k$ periods are already reserved by other nodes, it indicates a reservation blocking status. The probability associated with this scenario is denoted as $P_{bk}^k$.
\end{definition}

\subsection{The Proposed Hybrid Reservation Process}\label{The proposed hybrid protocol framework}

\begin{figure}[htbp]
\centering
\includegraphics[width=3.5in]{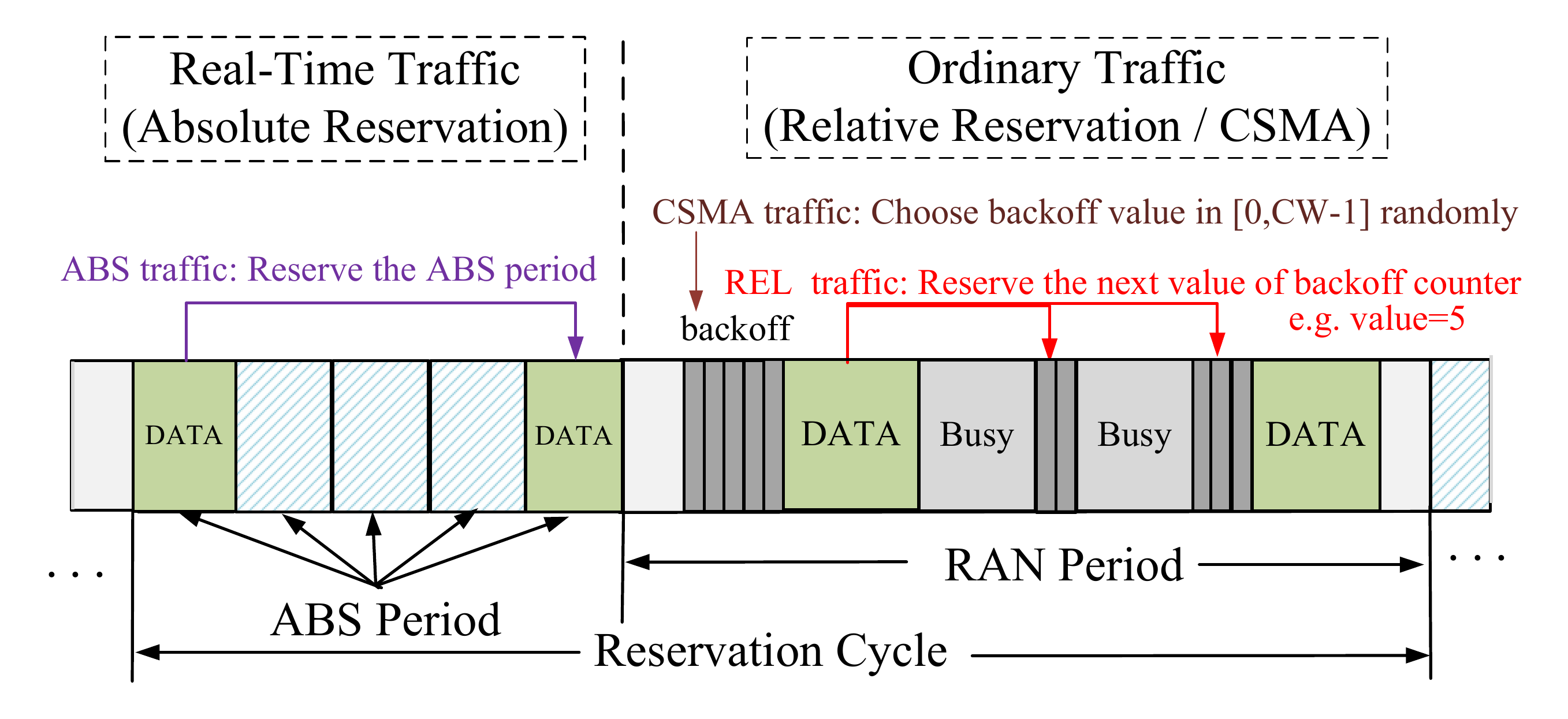}
    \caption{Proposed hybrid reservation process}
	\label{Reservation resource period}
\end{figure}

The proposed hybrid reservation protocol process is illustrated in Fig. \ref{Reservation resource period}. The channel resources are divided into the ABS periods and the RAN periods, which are reserved for real-time traffic and ordinary traffic, respectively. In the ABS periods, the real-time traffic utilizes the absolute reservation, which is defined as ABS traffic. Among the $k_{ord}$ types of traffic utilizing the RAN period, there are $k_{rel}$ types of traffic with relatively high-priority and $k_{cs}$ types of traffic with lower priority, using the relative reservation and CSMA, respectively.
The two types of traffic mentioned above are referred to as REL traffic and CSMA traffic, respectively.

The duration of each reservation cycle is predetermined by the protocol. When a node sends the RSV MPDU, it includes the information about the beginning time of the subsequent reservation cycle. Then the neighboring nodes synchronize the reservation cycle with this information. Let $\varepsilon_{abs}$ be the proportion of channel resources for absolute reservation, adjusted based on the real-time traffic volume in the network. In order to prevent excessive dominance of channel resources by high-priority traffic, the maximum proportion of resources allocated for absolute reservation is bounded by $\varepsilon_{abs}^{max}$.

The duration of each reservation in ABS periods is determined based on the queue status, subject to a maximum duration constraint $\varepsilon_{abs}^{max}$. To minimize the fragmentation, the reservation periods are tightly packed at the beginning of a reservation cycle. When initiating a reservation transmission, the node randomly selects a transmission period from $N$ idle periods that are not reserved by other nodes. The periods that can be chosen is donated by the Reservation Window (RSW) and the number of chosen periods is $N$. Each node determines the $N$ corresponding to the type of traffic, as per the algorithm outlined in Section \ref{Adaptive Absolute Reservation Parameter Algorithm}. If a packet has been waiting in the queue for a significant amount of time, a temporary shorter reservation offset time can be used and the average reservation offset time must still satisfy the calculation result.

During the $RAN$ period, the relative reservation protocol employed in this paper randomly selects a reserved backoff value from the future $N$ idle slots. In order to prevent collisions upon receiving a reserved backoff value $I_{BO}$, if the current backoff value of the neighboring node equals $I_{BO}$, then a new backoff value is selected randomly from the range $[I_{BO}-I_{ch}, I_{BO}+I_{ch}]$, where $I_{ch}$ is a constant used to avoid the collision. The algorithm presented in Section \ref{Adaptive Absolute Reservation Parameter Algorithm} is employed to dynamically adjust the size of the reservation window, which is updated dynamically to accommodate the current traffic volume.


This paper focuses on densely deployed scenarios. Previous research \cite{quqiaoXindao} has established a mathematical model that indicates nodes located nearby exhibit consistent channel perceptions. This observation suggests that the relative reservation mechanism can effectively reserve corresponding slots on densely deployed scenarios. In Section \ref{Problem of neighbor channel sensing capability}, we present simulation results to investigate scenarios where nodes are widely dispersed and possess inconsistent channel perceptions. The simulations show that the relative reservation mechanism outperforms the CSMA mechanism even in such scenarios.

\subsection{Adaptive Reservation Parameter Determination Algorithm}\label{Adaptive Absolute Reservation Parameter Algorithm}
\subsubsection{Process of Selecting Reservation Parameter}\label{Process of Reservation Parameter Decision  Process}
\ 
\newline 
\indent 
For absolute reservation and relative reservation mechanism, the adaptive reservation parameter determination algorithm is predicated on the idea that the reservation window sizes of different traffic are adjusted based on the traffic's $\alpha$ and the total reservation transmission ratio. For absolute reservation, the $\alpha$ of traffic is adjusted according to the maximum tolerance delay requirement and priority. 




\begin{algorithm}[H]
    \caption{Basic Reservation Process}
    \label{AL-Reservation Process}
    \begin{algorithmic}[1]
        \IF{Selecting the reservation period}
            \IF{For absolute reservation}
                \STATE{Choose a period in the reservation window}
					\IF {There is no reservation period to meet the delay requirement}
      				\STATE{Execute the soft reservation process}
					\ENDIF
            \ELSIF{For relative reservation}
                \STATE{Choose an idle backoff value in reservation window}
            \ENDIF
            \STATE{Indicate $\alpha$, the reserved period and the remaining tolerance delay of the traffic in RSV MPDU}
        \ELSIF{Receiving reservation information}
\STATE{Keep silent at the reservation period}
\STATE{Update the active neighbor node list}
\STATE{Recalculate the reservation window and $\alpha$ according to Algorithm \ref{ChangeAlfa} periodically}
        \ENDIF
    \end{algorithmic}
\end{algorithm}

The process for maintaining reservation information is outlined in Algorithm \ref{AL-Reservation Process}. When a node initiates a reservation, the reservation frame RSV MPDU indicates the reservation weight $\alpha$, the duration of the reservation period, and the remaining tolerance delay of the traffic. The reservation duration is determined by the traffic arrival rate which can be estimated by making use of the traffic prediction.  The traffic prediction has been extensively discussed in the researches \cite{Pro1,Pro2}. Then all the nodes maintain the active list of neighbor node reservation information.

\begin{algorithm}[htbp]
\caption{Adaptive Reservation Determination Parameter}
\label{ChangeAlfa}
\begin{algorithmic}[1]
\REQUIRE $\mathbb{A} = {\alpha_1, \ldots, \alpha_{k_{tot}}}, P_{RSV}^{sui}, \mathbb{L}={L_1^{c,min},\ldots,L_{k_{tot}}^{c,min}}, \mathbb{D}={d_{1}, \ldots,d_{k_{tot}}}$

\FOR{$i=1$ to $k_{low}$} 
\IF{$P_{bk}^{d_i} > Th_{bk}$} 
\STATE Change $P_{RSV}^{suit}$ to meet $Th_{bk}$
\ENDIF
\ENDFOR

${T_1,...,T_{k_{low}}}, \mathbb{N}^{abs} \gets$ ChangeReservationWindow($\mathbb{A}, P_{RSV}^{suit}$) 

\FOR{$i=1$ to $k_{low}$} 
\IF{$T_i > L_{min}$}
\STATE Change $\alpha_j$ to meet $L_{min}$
\ENDIF
\ENDFOR

\IF{there are still traffic that do not meet the requirement}
\FOR{$i=1$ to $k_{low}$}
\STATE $L_i^c \gets L_i^{typ}$
\ENDFOR
\ENDIF

$\mathbb{N}^{rel} \gets$ ChangeReservationWindow($\mathbb{A}, \varepsilon_{rel}$)

\ENSURE $\mathbb{A}, P_{RSV}^{sui}, \mathbb{N}^{abs}, \mathbb{N}^{rel}$
\end{algorithmic}
\end{algorithm}


\subsubsection{Absolute Reservation Parameter Determination Process}\label{Absolute Reservation Parameter Process}
\ 
\newline 
\indent The algorithm is comprised of three components: the correction of the total reservation transmission proportion, which is denoted by $P_{rsv}^{sui}$, adaptive adjustment of the reservation window, and the correction of reservation weight.

In the first step, the restriction condition on the $P_{rsv}^{sui}$ is modified based on the reservation blocking probability, as defined in Definition \ref{definition5}. A higher total reservation transmission proportion results in an increased likelihood of reservation blocking. Thus, the allowed reservation blocking quantity for each traffic is determined by considering the tolerance delay $d_i$. The $P_{rsv}^{sui}$ is adjusted until all traffic satisfies the requirement $P_{bk}^i\leq {Th}_{bk}$, where ${Th}_{bk}$ denotes the confidence level for the reservation blocking probability. This confidence level represents the maximum allowable probability for the occurrence of a reservation blocking state and is computed using Theorem \ref{theorem3}.

The next step is adaptive adjusting the reservation window, which is based on the reservation weight of each traffic, the reservation information from neighboring nodes, and the $P_{rsv}^{sui}$. Initially, it sets the search step size for the initial iteration. Firstly, for each reservation window size corresponding to the traffic in each search step, the total reservation transmission proportion and delay for each traffic are calculated using the mathematical model described in Section \ref{DifRsvModel}. Since the reservation window size is inversely related to the total reservation transmission proportion, the algorithm seeks reservation window sizes that satisfy the $P_{rsv}^{sui}$. Subsequently, based on the previous calculation results, if the delay of any traffic does not meet the requirements, the reservation weight of that specific traffic is adjusted until the requirements are satisfied.

In the third step, the process of adjusting the reservation weight calculates a suitable $\alpha$ for each traffic that meets the delay requirements, taking into account the current network conditions. Under heavy load or burst traffic scenarios, packets may experience prolonged waiting time in queues. To mitigate this issue, a more strict delay limit $L_i^{c,min}=\lambda L_i^{min}$ is employed as the criterion for reservation weight calculation.

The overall process is described in Algorithm \ref{ChangeAlfa}. Initially, the reservation window size for each traffic is calculated by the original reservation weight based on the traffic priority. This calculation is performed using the reservation window adjustment process. If the delay requirements of all real-time traffic cannot be satisfied after adjusting the reservation weight, then the algorithm is rerun utilizing the typical delay as the constraint. The typical delay is denoted as $L_i^{typ}=(L_i^{min}+L_i^{max})/2$. 

When the algorithm determines that the reservation weight for a particular type of traffic needs to be adjusted, it signifies a heavy load condition. In such situation, the maximum reservation offset time that fulfills the QoS requirement for this traffic will be utilized consistently.


\subsubsection{Relative Reservation Parameter Determination Process}\label{Relative Reservation Parameter Process}
\ 
\newline 
\indent Regarding the relative reservation mechanism, the transmission ratio for relative reservation mechanism is denoted as $\varepsilon_{rel}$. That is, the relative reservation window should be configured to fulfill that the proportion of relative reservation transmission duration to the total RAN periods is $\varepsilon_{rel}$.

The calculation of the relative reservation transmission ratio relies on the reservation window of different priority traffic. The mathematical model introduced in Section \ref{DifRsvModel} offers a means to calculate the relative reservation transmission ratio based on the specified reservation parameters of each traffic and the number of nodes. As a result, the reservation window adjustment process outlined in Algorithm \ref{ChangeAlfa} is utilized to determine the relative reservation parameters using the parameters $\alpha$ and $\varepsilon_{rel}$.
\subsection{Improvement for Fragmentation Problem}\label{Improvement for Fragmentation Problem}
\begin{figure}
\centering
\includegraphics[width=2.7in]{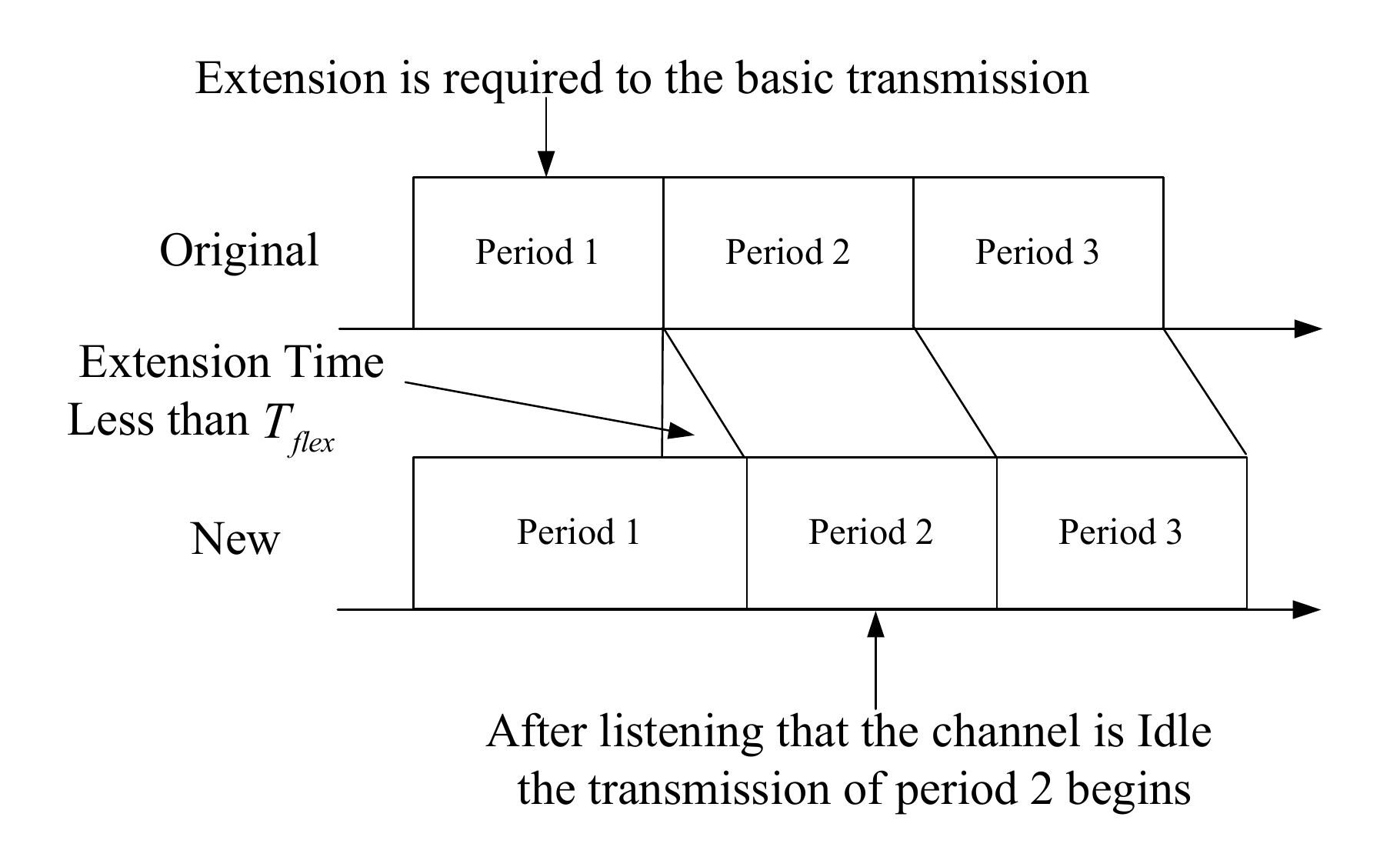}
    \caption{ABS period flexible processing}
	\label{ABS period flexible processing}
\end{figure}
In scenarios where all nodes are within the two-hop range, an enhanced fragmentation optimization mechanism can be employed to improve efficiency. This mechanism introduces a flexible reservation duration, denoted as $T_{fle}$, within the ABS period. This duration is significantly shorter than the basic transmission duration. When a node gains the channel transmission opportunity and initiates the transmission, if the transmission extends into the next ABS period and the additional reservation time required is less than $T_{fle}$, the node is permitted to utilize that portion of time for the transmission. During the reservation in ABS periods, the period ID indicates the specific period in the current reservation cycle. Neighbor nodes determine whether the transmitting node's period has arrived by evaluating the ID provided by the transmitting node. 

\subsection{Soft Reservation}\label{Soft Reservation}

During the ABS period, the unreserved periods permit channel access through the CSMA once the corresponding time is reached. In the idle ABS periods, if there exists some traffic in the queue but no avaliabe reservation opportunities, the nodes can employ a reduced contention window for channel access, and it is same to the case that the waiting time of real-time traffic in the queue has exceeded the threshold.

For networks with fluctuations in network load, the following situation can arise. During periods with low network load, existing reservation resources can adequately meet transmission needs, the nodes utilized smaller reservation windows according to the algorithm. If some nodes unexpectedly generate a large number of real-time traffic, there are not enough nearest reservation periods to meet the delay demand of the newly arrived traffic.

To address these issues, this paper presents an enhanced Soft-Reservation mechanism. In research \cite{juedui2,juedui4}, it has investigated Soft-Reservation mechanisms that employ reservation cancellation frames to release the allocated channel resources when the reservation source node does not have sufficient traffic to transmit. This paper adds a negotiated soft reservation mechanism based on it.

In the proposed Soft-Reservation mechanism, when the network load is light, the node $n_{k1}$ can reserve the period (e.g. period $t_{les}$) with less offset time than the maximum tolerance delay. These reserved periods are donated by Soft-Reservation periods. Furthermore, the latest period $t_{lat}$ that satisfies the maximum tolerance delay of traffic is also specified in the RSV MPDU. When another node $n_{k2}$ is selecting the reservation period but no available periods meet the requirement. Then it selects a period (e.g. period $t_{les}$)  from the recorded Soft-Reservation list. If there is any unreserved period $t_{idl}$ before the latest period $t_{lat}$ corresponding to $t_{les}$, the node $n_{k2}$ will negotiate to reserve $t_{les}$ while node $n_{k1}$ will change to occupy the $t_{idl}$. According to the above mechanism, the traffic of both nodes can meet the delay requirements. 


\section{Performance Analysis}\label{PERFORMANCE ANALYSIS}
\subsection{Fundamental Assumption}\label{Fundamental assumption}


Regarding the impact of non-ideal channels, this paper addresses the decoding failure of both data frames and RSV MPDUs, which carry reservation information. Let $P_e$ be the probability of a data frame being failed decoded by the receiver. Note that the RSV MPDU employs a more robust encoding mode. Thus, we assume that the probability of successful decoding of the RSV MPDU is $h_r=1-(P_e)^{G_r}$, where $G_r\geq 1$.

The analysis primarily centers on the performance under conditions of saturated traffic. In such a scenario, the arrival of new data is immediately followed by successful transmissions and a new reservation will be initiated. It is assumed that all nodes are located within a fully connected region. The above hypothesizes have been widely adopted in related researches which can refer to \cite{intro8, Paperbianqi}.

Let $T^{rsv}_{suc}$ and $T^{rsv}_{fai}$ be the average duration of successful and failed reservation transmission, respectively. Similarly, let $T_{suc}^{cs}$ and $T_{fai}^{cs}$ be the average duration of successful and failed CSMA transmissions, respectively.

This section assumes that each node is capable of generating only one type of traffic. Nodes using absolute reservation, relative reservation and CSMA are denoted as ABS nodes, REL nodes and CSMA nodes, respectively. The number of three types of nodes are respectively denoted as $m_{abs}$, $m_{rel}$ and $m_{cs}$. The total number of network nodes is denoted as $m$.

\subsection{Differentiated Reservation Model}\label{DifRsvModel}

The proposed model considers the varying sizes of reservation windows for all nodes, which calculates the probability of different reservation periods selected by each node and the reservation transmission ratio of each node. For the relative reservation mechanism, each period corresponds to a slot, which can be idle or used for transmission. And for the absolute reservation mechanism, each period represents a transmission process, encompassing transmissions initiated by absolute reservation mechanism or CSMA. The model performs a separate calculation for each reservation mechanism. In this section, we define the number of nodes using the corresponding reservation mechanism to be solved as $m$.

When a tagged node $n_k (1\le k\le m)$ initiates a reservation transmission, it randomly chooses an idle period from the $N_k$ idle periods as its reservation period. Assuming it is in a steady state when a node is selecting the reservation period,  it is possible that the future reservation periods have been occupied by other $m-1$ nodes. Consequently, the node can reserve a period from the future $C_k$ transmission periods, where $C_k=N_k+m-1$.

\begin{definition}\label{definition1}
Left reservation window and right reservation window


\begin{figure}
\centering
\includegraphics[width=3.5in]{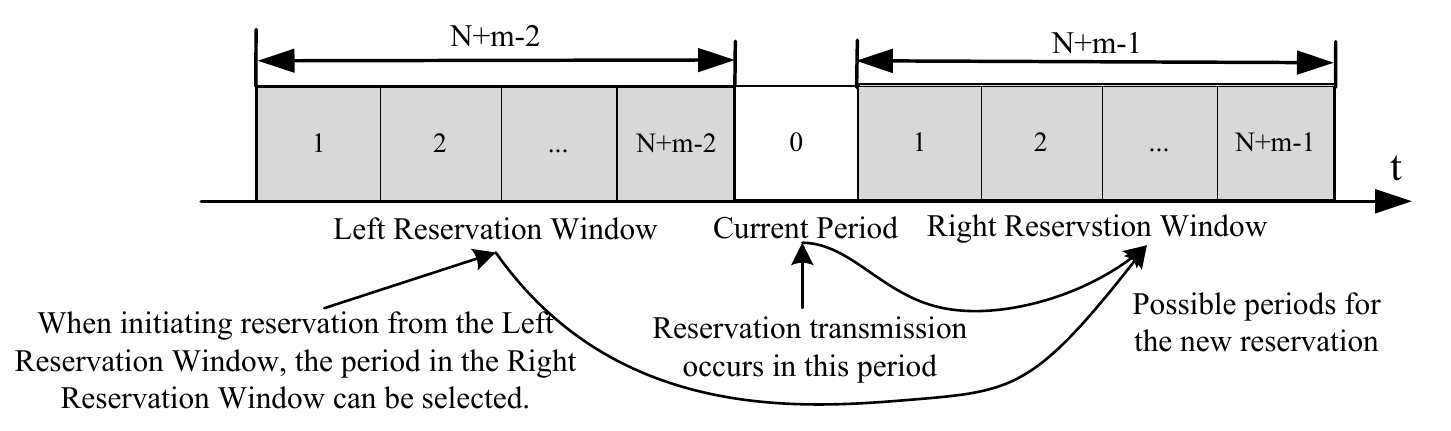}
\caption{Reservation window diagram}
\label{ReservationWindow}
\end{figure}

As shown in Fig. \ref{ReservationWindow}, the range of available periods for reservation is defined as the right reservation window with the size $C_k$. Additionally, we define the left reservation window, which allows the selection of a period in the right reservation window upon initiating the reservation transmission. Therefore, the size of the left reservation window is $C_k-1$. The $j$-th period in the left reservation window is denoted by $T_{all}^{j}(j\in[1,C_k-1])$, while the $j$-th period in the right reservation window is denoted as $T_{r}^{j}$ or $T_{all}^{C_k-1+j} (j\in[1,C_k])$.
\end{definition}

According to Fig. \ref{ReservationWindow}, the probability of the later period in the right reservation window being reserved by other nodes decreases. The fundamental concept of the model is as follows: As the current period progresses from the start to the end of the left reservation window, the probability of selecting each period as a reservation period are calculated sequentially. Then the probabilities of each period in the right reservation window being already reserved are determined. Finally, the probabilities of each period in the right reservation window being chosen as a reservation period are computed.


\begin{definition}\label{definition2}
 Active reservation probability and Idle period probability

Given the current period is $T_{all}^{t} (j\in[1,2C_k-1])$, the probability that node $n_k$ has a reservation for $T_{all}^{i}$ ($i\in[t+1,t+C_k]$) is defined as the active reservation probability, represented by $R_t^{k,i}$. The probability that $T_{all}^i$ is idle, i.e., has not been reserved by any node, is defined as the idle period probability, represented by $S_t^i$. It is obvious that
\begin{equation}
S_t^i=1-\sum_{k=1}^mR_t^{k,i}=1-R_t^i.
\label{eqSti}
\end{equation}
\end{definition}

\begin{definition}\label{definition3}
Active Reservation Period Range Probability

Let $P_{rsv,k}^{(i,j),d} (i\in[1,2C_k-1],j\in[i-1,2C_k-1],d\in[0,j-i+1],k\in[1,m])$ denote the probability that there are $d$ periods reserved by other nodes from $T_{all}^i$ to $T_{all}^j$. 
\end{definition}

In order to analyze the probability of node $n_k$ choosing period $T_{r}^j$ as the reservation period in the right reservation window, the following Theorem \ref{Theo1} are presented. 

\begin{theorem}
\label{Theo1}
For a given node $n_k$ which is selecting the reservation period, let $T_{all}^s(s\in[1,2C_k-1])$ be the current period and $T_{all}^d(d\in[s+1,s+C_k])$ be a period in the right reservation window of $T_{all}^s$. Then the probability of selecting $T_{all}^d$ as the reservation period is obtained as
\begin{equation}	
P_k(T_{all}^s,T_{all}^d)=\frac{S_s^d}{\sum_{z=s+1}^{s+C_k}S_s^z}.
\label{equtheorem1}
\end{equation}
\label{theorem1}
\end{theorem}
\begin{proof}
The total probability of selecting a period from the right reservation window is
\begin{equation}	
\sum_{d=s+1}^{s+C_k}{P_k(T_{all}^s,T_{all}^d)=1}.
\label{equt1Sum1}
\end{equation}

Since the probability of selecting a period as reservation period is proportional as the reservation period to its idle probability, we have
\begin{equation}	
\begin{split}
P_k(T_{all}^s,T_{all}^{s+1}):\cdots:P_k(T_{all}^s,T_{all}^{s+C_k})=S_s^{s+1}:\cdots:S_s^{s+C_k}.
\label{Ptk1T}
\end{split}
\end{equation}

Combining Eq. (\ref{equt1Sum1}) and Eq. (\ref{Ptk1T}), Eq. (\ref{equtheorem1}) holds immidiately.
\end{proof}

\begin{definition}\label{definitionHasRsvSta}
Reservation Transmission Period, In-Reservation Status, and Non-Reservation Status

For each period, the probability 
	for the reservation transmission period of a tagged node $n_k(k\in[1,m])$ is represented by $P_{IR}^k$. After the node $n_k$ has successfully initiated a reservation for the future reservation period, it is considered to be in the in-reservation status, denoted by the probability $P_{has}^{rsv,k}$. Conversely, after the RSV MPDU of node $n_k$ decoding fails or collision occurs and tries to obtain the transmission opportunities through CSMA, it is in a non-reservation status, and the corresponding probability is $P_{no}^{rsv,k}$. And the average probability of all nodes for these two probabilities is $P_{has}^{rsv}$ and $P_{no}^{rsv}$, respectively.
\end{definition}
Theorem \ref{theorem2} is used to obtain the probability of each period being reserved by the node $n_k$, which is denoted by $P_{IR}^{k}$.
\begin{theorem}
Assuming that the tagged node $n_k(k\in[1,m])$ chooses $T_{r}^j(j\in[1,C_k])$ as the reservation period with the probability of $P_{RSV}^j$, and $P_{cs}^{suc}$ represents the probability of successfully initiating a transmission through CSMA. Then we have
\begin{equation}	
\label{HasRsvStatusP}
		\left\{
		\begin{aligned}
&P_{IR}^k=\frac{1}{(1-h_r)/P_{cs}^{suc}+\sum_{k=1}^{C_k}{\sum_{j=k}^{C_k}{P_{RSV}^j}}}\\
			&P_{no}^{rsv,k}=\frac{1-h_r}{(1-h_r)+P_{cs}^{suc}\sum_{k=1}^{C_k}{\sum_{j=k}^{C_k}{P_{RSV}^j}}}\\
			&P_{has}^{rsv,k}=1-P_{no}^{rsv,k}
		\end{aligned}
		\right.
.
\end{equation}
\label{theorem2}
\end{theorem}
\begin{proof}
Consider the stochastic process $b^k(t)$ that represents the remaining reservation periods for a given node $n_k$ at time $t$. The remaining periods mean the number of periods that needs to be waited before reaching the reservation period. Specifically, when $b^k(t)=1$, it means that the node's reservation period will reach in the next period. Then if the reservation information is decoded successfully by the destination node, it will initiated a new reservation and select the reservation period. When $b^k(t)=0$, it signifies a non-reservation state resulting from the failed reservation information decoding or collision. In the non-reservation state, the node will re-enter the in-reservation state after a successful CSMA transmission. Therefore, as shown in Fig. \ref{LeftReservationPeirod}, for a node with the reservation window size $N$, the one-step transition probabilities can be expressed as follows.
	\begin{eqnarray}
		\left\{
		\begin{aligned}
			&P \left\{ j \vert j+1 \right\} =1 \\ 
			&P \left\{ j \vert 1 \right\} =h_rP_{RSV}^j\\ 
			&P \left\{ j \vert 0 \right\} =P_{cs}^{suc}P_{RSV}^{j}\\
			&P \left\{ 0 \vert 1 \right\} =1-h_r&& 
		\end{aligned}
		\right.
,
	\end{eqnarray}
 where $j\in [1,C_k]$.

Assuming the probability of remaining reservation periods for node $n_k$ in the steady state is $b_i^k$, the following equations hold:
	\begin{eqnarray}
\label{EqForTheo2}
		\left\{
		\begin{aligned}
			&b_0^k=\frac{(1-h_r)b_1^k}{P_{cs}^{suc}}\\
			&b_i^k=(b_0^kP_{cs}^{suc}+b_1^kh_r)\sum_{j=i}^{C_k}{P^j_{RSV}}
		\end{aligned}
		\right.
\label{WentaiCsmab}
,
	\end{eqnarray}

where $i\in [1,C_k], k\in[1,m]$.
By solving the normalized condition $\sum_{i=1}^{C_k}{b_i^k}=1$, the probability of each state is
\begin{equation}	
\label{BiWenTai}
b_i^k=\frac{\sum_{j=i}^{C_k}P_{RSV}^{j}}{(1-h_r)/P_{cs}^{suc}+\sum_{z=1}^{C_k}\sum_{j=z}^{C_z}P_{RSV}^{j}}.
\end{equation}
\begin{figure}
    \centering
    \includegraphics[width=0.3\textwidth]{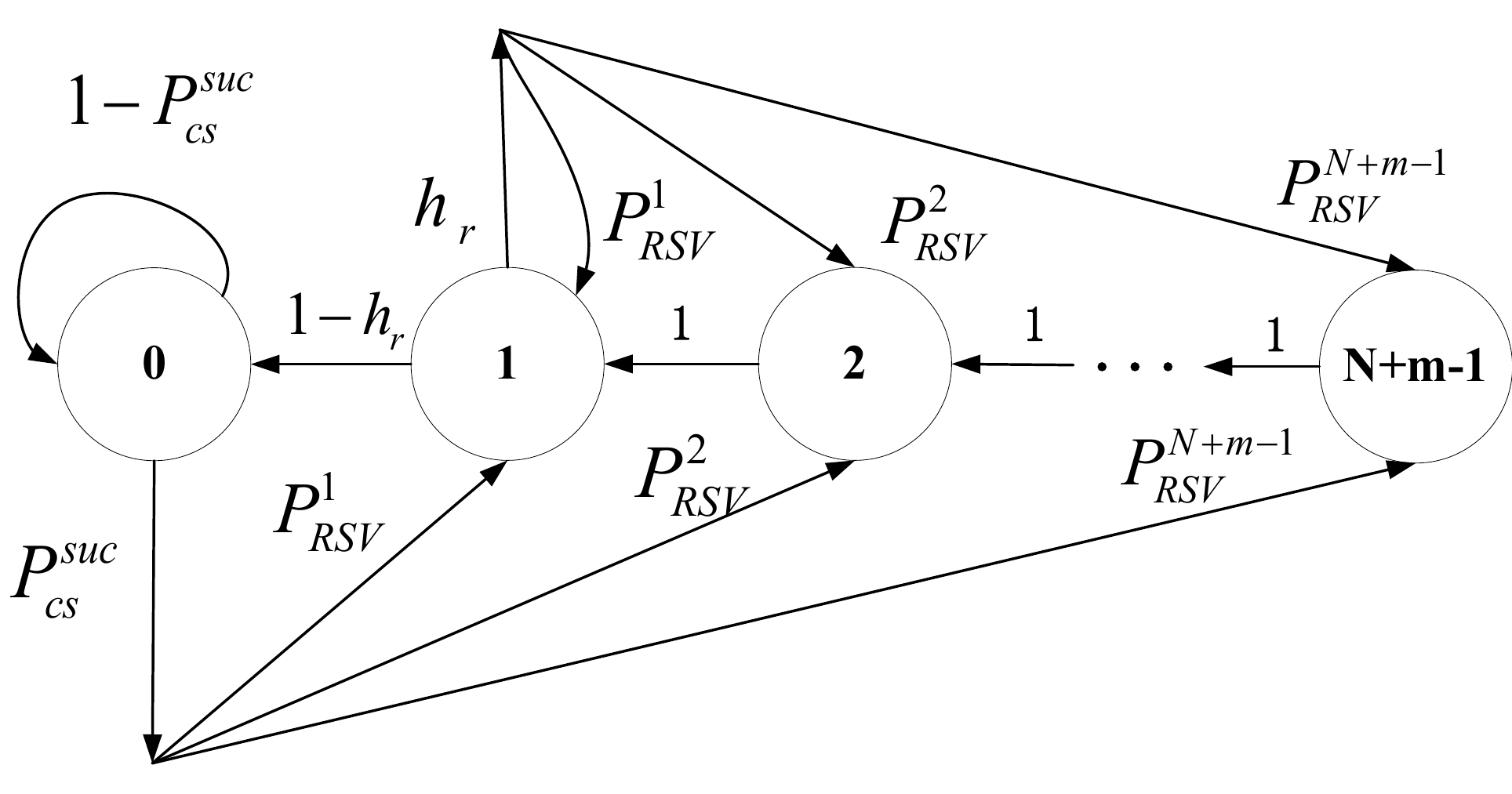}
    \caption{Status transition diagram of the remaining reservation periods}
	\label{LeftReservationPeirod}
\end{figure}

Since $b^k(t)=1$ and $b^k(t)=0$ corresponds to the reservation transmission period and non-reservation status, respectively. Therefore, substituting $i=1$ and $i=0$ into the Eq. (\ref{BiWenTai}), then Eq. (\ref{HasRsvStatusP}) holds.
\end{proof}

Furthermore, the reservation blocking probability can be obtained by using the Theorem \ref{theorem3} and Lemma \ref{lemma1}.

\begin{lemma} 
Given the current period $T_{all}^i(i\in[1,2C_k-1])$, the probability of the period $T_{all}^j(j\in[i-1,2C_k-1])$ in an active reservation state is $R_{i}^j$, then we have
\begin{equation}	
P_{rsv,k}^{(i,j),d}={P_{rsv,k}^{(i,j-1),d-1}R_{i}^j+P_{rsv,k}^{(i,j-1),d}(1-R_{i}^j)},
\end{equation}
\label{lemma1}
where $j\ge i-1$.
\end{lemma}
\begin{proof}
The above probability corresponds to two cases.
In the first case, there is a total of $d-1$ reserved periods in the previous $j-1$ periods and the period $T_{all}^j$ is reserved by other nodes.
In the second case, there is a total of $d$ reserved periods in the previous $j-1$ periods and the period $T_{all}^j$ remains unreserved. 
The initial value of the cumulative formula is $P_{rsv,k}^{(i,i-1),0}=1$, that is, the number of reserved periods is $0$ when the total number of periods is $0$.
\end{proof}
\begin{theorem}

Given the reservation blocking quantity $d$ (defined in Definition \ref{definition5}). In the first $d$ periods of the right reservation window, if the probability of $j$ periods being occupied is expressed as $P_{rsv,k}^{(1,d),j}(j\in[0,d],k\in[1,m])$, then the reservation blocking probability is $P_{bk}^d=P_{rsv,k}^{(1,d),d}$.
\label{theorem3}
\end{theorem}
\begin{proof}
Based on Definition \ref{definition5} and Lemma \ref{lemma1}, the reservation blocking occurs when all $d$ periods in the beginning of the right reservation window are reserved by other nodes.
\end{proof}

\begin{figure}[!t]
\renewcommand{\algorithmicrequire}{\textbf{Input:}}
\renewcommand{\algorithmicensure}{\textbf{Output:}}
\begin{algorithm}[H]  
  \caption{Differentiated Reservation Parameter Model}  
  \label{alg:DifRsv}  
\begin{algorithmic}[1]
 \REQUIRE{$R_0^{k,j}=P_{has}^{rsv,k}/C_{max}$, $P_{rsv}^{(i,i-1),0}=1$ $(k\in[1,m],j\in[1,N_{max}+m-2])$}

\FOR{$s=1$ to $N_{max}+m-2$} 
        \FOR{$k=1$ to $m$}
            \FOR{$d=s+1$ to $s+C_k$}
                \STATE{$S_{s}^{d}=1-\sum_{z=1}^m{R_{s}^{z,d}}$}

                \STATE{$P_{k}(T_{all}^s,T_{all}^d)={S_{s}^{d}}/ \sum_{z=1}^{C_k} {S_{s}^{z}}$ }
                \STATE{$R_{s}^{k,d}=R_{s-1}^{k,d}+R_s^{k,s}P_{k}(T_{all}^s,T_{all}^d)$}
            \ENDFOR
        \ENDFOR
    \ENDFOR
    \FOR{$k=1$ to $m$}

        \FOR{$j=1$ to $C_k$}
            \STATE{$P_{RSV}^{k,j}=S_{cur}^{j}/\sum_{z=1}^{C_k}{S_{cur}^z}$}
\STATE{$b_j=\sum_{i=j}^{C_k}{P_{RSV}^{i}}/\sum_{i=1}^{C_k}{\sum_{z=i}^{C_k}{P_{RSV}^z}}$}        
\ENDFOR
        
        \STATE{$P_{IR}^{k}=b_1, P_{no}^{rsv,k}=b_0$}
    \ENDFOR
\STATE{$P_{RSV}^{j}=\sum_{k=1}^m{P_{RSV}^{k,j}}$, $P_{IR}^{tot}=\sum_{k=1}^m{P_{IR}^{k}}$}
\ENSURE{$P_{RSV}^{j}$, $b_j$, $P_{no}^{rsv,k}$ and $P_{IR}^{tot}$}
\end{algorithmic}
\end{algorithm}
\end{figure}

Let $P_{has}^{rsv,k}$ be the probability of node $n_k(k\in[1,m])$ in the in-reservation state, and $R_t^{k,j}(t\in[1,2C_k-1],j\in[t+1,t+C_k])$ be the probability that node $n_k$ has reservation for $T_{all}^j$ at the current period $T_{all}^t$.	Since the reservation period selection is memoryless, an approximately equal initial active reservation probability is assigned to each period in the left reservation window. Therefore, we have
\begin{equation}	
\label{R000}
R_0^{k,j}=P_{has}^{rsv,k}/C_k.
\end{equation}

Suppose that at period $T_{all}^t(t\in[1,2C_k-1])$, node $n_k(k\in[1,m])$ has a probability of $P_k(T_{all}^t,T_{all}^j)$ for selecting period $T_{all}^j(j\in[t+1,t+C_k])$ as the reservation period. This selection increases the probability of period $T_{all}^j$ in the active reservation state. The increased probability is obtained by multiplying the probability of period $T_{all}^t$ in the active reservation state with selecting period $T_{all}^j$. Consequently, we have
\begin{equation}	
\label{leijiagongshi}
R_{t+1}^{k,j}=R_t^{k,j}+R_t^{k,t}\times P_t^k(T_{all}^t,T_{all}^j).
\end{equation}

By using Lemma \ref{lemma1} and Theorem \ref{Theo1}, combining with the Eq. (\ref{leijiagongshi}) sequentially,  we present the overall calculation process of the model, as shown in Algorithm \ref{alg:DifRsv}. It sequentially calculates the probabilities of each period in the right reservation window being reserved as the timeline progresses from the beginning to the end of the left reservation window. Then the probability of each period in the right reservation window being reserved can be obtained.
Therefore, the probability of each transmission period being reserved can be obtained. 

\subsection{Backoff Model}\label{Backoff Model}
For the RAN period, there are three types of slots: idle slots, relative reservation slots and CSMA transmission slots. This section introduces the ratio of these slots and the probability of successful transmission through relative reservation and CSMA. Moreover, we give some formulas in this section, which can be used to calculate the performance of CSMA nodes and REL nodes using CSMA in the non-reservation status.

According to the reservation model introduced in Section \ref{DifRsvModel} and the reservation window sizes used by each node, we can obtain the probability that a given period has been relatively reserved by node $n_k(k\in[1,m_{rel}])$, denoted as $P_{rel}^k$. The total probability that a period has been relatively reserved by any node is $P_{rel}^{tot}=\sum_{k=1}^{m_{rel}}{P_{rel}^k}$.


For CSMA transmission opportunities with the contention window $CW_i$, where $i$ represent the backoff stage, the node randomly choose a backoff counter value from $CW_i$ future unreserved slots. The range of selectable backoff counter values in the $i$-th backoff stage is $[0,W_i^{'}]$, where $W_i^{'}=W_i+m_{rel}-1$ is the revised contention window. According to Theorem \ref{theorem1} and the probability that the slot has been reserved within the available range, we can obtain the probability of selecting the backoff value $j$ in the $i$-th backoff stage, denoted as $P_{bo}^{i,j}$, where $j\in[0,W_i^{'}]$.
\begin{figure}[htbp]
\centering
\includegraphics[width=3in]{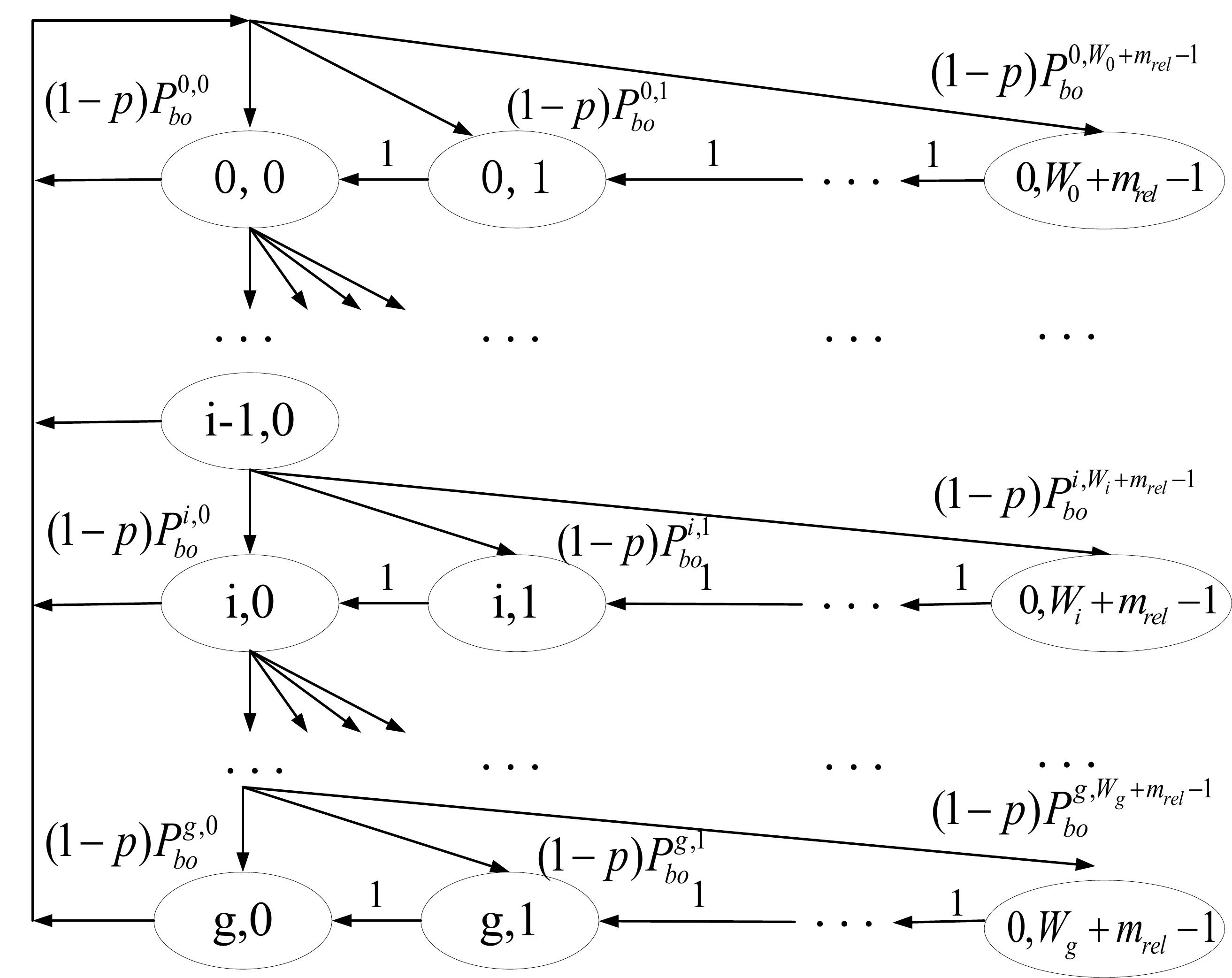}
    \caption{Corrected Markov chain model for the backoff value}
	\label{fig:bianqi}
\end{figure}



The subsequent calculation is based on the classic model proposed by \cite{Paperbianqi}. Under the condition that some slots are occupied by reservations, the state transition diagram of backoff value has been modified according to the relative reservation rules, which is as shown in Fig. \ref{fig:bianqi}, where $p$ denotes the probability of failed CSMA transmission, and $g$ signifies the maximum backoff stage. In this model, $s(t)$ represents the stochastic process of the backoff stage and $b(t)$ represents the value of backoff counter. Consequently, the one-step transition probability of the corresponding Markov chain can be obtained as
	\begin{eqnarray}
\label{yibuzhuanyi}
		\left\{
		\begin{aligned}
			&P \left\{ i,k \vert i,k+1 \right\} =1  ,& k\in (0,W_i^{'}),\ &i\in(0,g) \\
			&P \left\{ 0,k \vert i,0 \right\} =(1-p)P_{bo}^{0,k} ,& k\in (0,W_i^{'}],\ &i\in(0,g)\\
			&P \left\{ i,k \vert i-1,0 \right\} =pP_{bo}^{i,k} ,& k\in (0,W_i^{'}],\ &i\in(1,g)\\
			&P \left\{ i,k \vert i,0 \right\} =pP_{bo}^{i,k},&  k\in (0,W_i^{'}],\ &i=g\\
		\end{aligned}
		\right.
.
	\end{eqnarray}

Let $h_{i,k}= \lim\limits_{t\to+\infty}P\left\{s(t)=i,b(t)=k\right\} $. We refer to Eq. (\ref{bianqi-1}) and Eq. (\ref{bianqi-2}) in \cite{Paperbianqi}:
	\begin{eqnarray}
\label{bianqi-1}
		h_{i,0}=\left\{
		\begin{aligned}
			&p^ih_{0,0}  ,& i\in(0,g) \\
			&p^ih_{0,0}/(1-p) ,& i=g
		\end{aligned}
		\right. ,
	\end{eqnarray}
and
\begin{equation}
 \tau = h_{0,0}/(1-p)
\label{bianqi-2}	
,
\end{equation}
where $\tau$ represents the probability of transmission at each slot for a tagged node.

The nodes that access the channel using CSMA consist of CSMA nodes and REL nodes that have not obtained reservation transmission opportunities. Let us define these nodes as competing nodes and represent their average count as $m_{con}$. Since the average probability of all the REL nodes in the non-reservation status is $P_{no}^{rsv}$, then the expected count
of competing nodes is
\begin{equation}	
	\label{Emcount}
E[m_{con}]=m_{cs}+m_{rel}P_{no}^{rsv}
.
\end{equation}

The collision occurs when multiple nodes initiate transmission via CSMA in an unreserved slot. The probability of it is denoted by $\tau^{'}=\tau/(1-P_{rel}^{tot})$. Thus, we have
\begin{equation}	
	\label{cacPc}
P_c=1-(1-\frac{\tau}{1-P_{rel}^{tot}})^{m_{con}-1},
\end{equation}
\begin{equation}	
p=1-(1-P_c)(1-P_e).
\end{equation}

Let $P_{tr}$ denote the probability that at least one node tries to initiate transmission through CSMA in an unreserved slot, and $P_s$ represents the probability of a successful CSMA transmission. The success condition is that only one node transmits in this slot, under the condition that at least one node transmits. Therefore, the following formula holds:
\begin{equation}	
P_{tr}=1-(1-\tau/(1-P_{rel}^{tot}))^{m_{con}},
\end{equation}
\begin{equation}	
P_s=\frac{m_{con}\tau}{(1-P_{rel}^{tot})}\cdot\frac{(1-\tau/(1-P_{rel}^{tot}))^{m_{con}-1}}{P_{tr}}.
\end{equation}

When it is in the steady state, the following equation is obtained according to the Eq. (\ref{yibuzhuanyi}):
\begin{equation}	
h_{i,k}=\sum_{j=k}^{W^{'}_i}{P_{bo}^{i,j}}h_{i,0},
\end{equation}
where $k\in(0,W_i-1)$.

According to the normalized condition, we have
\begin{equation}
	\label{CsmaAll1}	
\sum_{i=0}^mp^ih_{0,0}\sum_{k=0}^{W_i^{'}}{\sum_{j=k}^{W_i^{'}}{P_{bo}^{i,j}}}=1,
\end{equation}
where $P_{bo}^{i,j}$ is a constant obtained by the model introduced in Section \ref{DifRsvModel}. By combining Eq. (\ref{cacPc}--\ref{CsmaAll1}), we have
\begin{equation}	
\tau=\frac{1}{(1-P_{rel}^{tot})(\sum_{i=0}^{g-1}{p^i}+\frac{p^g}{(1-p)})(\sum_{k=0}^{W_i^{'}}{ \sum_{j=k}^{W_i^{'}}{P_{bo}^{i,j}}  })}.
\end{equation}

The formulas in this section are used to calculate the performance of CSMA nodes and REL nodes using CSMA in the non-reservation status.
\subsection{Calculation of RAN Period Transmission Probability}\label{RanPeriodTrans}
For the REL node, the transmission process includes the reservation transmission process at the in-reservation status and the CSMA transmission process at the non-reservation status.

In the RAN period, the probability of a collision-free transmission process initiated by the REL node is denoted as $P_{ran}^{rel}$. Let $P_{rel}^{rsv}$ be the probability of transmission initiation through the reservation, and $P_{rel}^{cs}$ be the probability of successful transmission initiated through CSMA for the REL node. In the relative reserved slot, the relative reservation transmission process will certainly occur. In the unreserved slot, there is a probability $P_{tr}$ for the occurrence of a CSMA transmission process, which can be initiated by a REL node at the non-reservation status or a CSMA node. Therefore, the proportion of transmission initiation by reservation is obtained by
\begin{equation}	
P_{rel}^{rsv}=\frac{P_{rel}^{tot}}{P_{rel}^{tot}+(1-P_{rel}^{tot})P_{tr}}.
\label{PrelRsv}
\end{equation}

For the contention nodes, the ratio of REL nodes is 
\begin{equation}	
P_{con}^{rel}=\frac{m_{rel}P_{no}^{rsv}}{m_{con}}.
\end{equation}

Therefore, for a given transmission, the probability that it is a successful transmission process initiated by the REL node through CSMA can be obtained by
\begin{equation}	
P_{rel}^{cs}=(1-P_{rel}^{rsv})P_sP_{con}^{rel}.
\end{equation}

Then for a given transmission in the RAN period, the probability that it is a successful transmission initiated by the REL node is
\begin{equation}	
P_{ran}^{rel,suc}=(P_{rel}^{rsv}+P_{no}^{rsv}P_{rel}^{cs})(1-P_e).
\end{equation}

The average transmission duration of REL nodes and CSMA nodes is as follows

\begin{align}
	\left\{
	\begin{aligned}
&T_{ran}^{cs}=u_{bo}\sigma+(1-P_s(1-P_e))T_{fai}^{cs}+P_s(1-P_e)T_{suc}^{cs}\\
&T_{ran}^{rel}=P_{ran}^{rel,suc}T_{suc}^{rsv}+(P_{no}^{rsv}P_{rel}^{cs})(1-P_s(1-P_e))T_{fai}^{rsv}\\
	\end{aligned}
	\right.
,
\end{align}
where $\sigma$ represents the duration of a slot.
Therefore, the average transmission duration for RAN is obtained by
\begin{equation}	
T_{ran}^{tot}=P_{ran}^{rel}T_{ran}^{rel}+P_{rel}^{cs}T_{ran}^{cs}.
\end{equation}

For a given REL node, the probability of successful transmission in each slot through CSMA is obtained by
\begin{equation}	
P_{cs}^{suc}=\frac{P_{tr}P_sh_r}{m_{con}}.
\end{equation}



Let $P_{ran}$ denote the ratio of the RAN period, and $P_{ST}^{rel,k}$ denote the probability of a successful transmission initiated by node $n_k$ in a certain transmission. If node $n_k$ is a REL node, the probability of successful transmission initiated by the node is obtained by
\begin{equation}	
P_{ST}^{rel,k}=P_{ran}P_{ran}^{rel,k}(1-P_e).
\end{equation}

The probability of a successful transmission for a tagged REL node $n_k$ using CS A is obtained by
\begin{equation}	
P_{ST}^{cs,k}=P_{ran}(1-P_{rel}^{rsv})P_{tr}P_s(1-P_e)/m_{con}.
\end{equation}

For a given slot in the RAN period, the probability that it is busy is denoted as $P_{ran}^{bs}$ and it can be calculated by 
\begin{equation}	
P_{ran}^{bs}=P_{rel}^{tot}+(1-P_{rel}^{tot})P_{tr}.
\end{equation}

Given the reserved backoff value selected by a node is $s$. The probability of $z$ busy slots being occupied by other nodes before the backoff value changes to $0$ is denoted as $P_{bet}^{s,z}$, and the corresponding formula is expressed as
\begin{equation}	
P_{bet}^{s,z}=C_s^z(P_{ran}^{bs})^z(1-P_{ran}^{bs})^{s-z},
\end{equation}
where $z\in[0,s]$.

Therefore, the probability of other nodes initiating $z$ transmissions between two reservation transmissions of node $n_k$ can be obtained by
\begin{equation}	
\label{PBETZ}
P_{bet}^{z}=\sum_{s=z}^{C_k}P_{RSV}^s.
\end{equation}

 Subsequently, the corresponding probability density function of the delay between two reservation transmissions can be computed based on Eq. (\ref{PBETZ}).	

\subsection{Calculation of ABS Period Transmission Probability}\label{AbsTransCac}
This section aims to calculate the successful transmission probability in the ABS period.


 Firstly, we can obtain the absolute reservation transmission probability $P_{abs}^k$ in the ABS period of a given node $n_k(k\in[1,m])$ and the total ratio of reservations during the ABS period $P_{abs}^{tot}$ according to the model described in Section \ref{DifRsvModel}.
The absolute reservation ratio directly corresponds to the ratio of reservation transmission processes. Therefore, the probability of successful transmission of the RSV MPDU, initiated by CSMA, at the non-reservation status, in each period for a node at the non-reservation status can be obtained by	
\begin{equation}	
	\label{AbsCsmaSuc}
P_{cs}^{suc}=\frac{(1-P_{abs}^{tot})P_sh_r}{m_{abs}}.
\end{equation}

When considering the un-ideal channel, we need to iteratively input the above parameters into the model described in Section \ref{DifRsvModel} to recalculate the reservation transmission ratio.

During the ABS periods, the nodes can transmit via absolute reservation or CSMA. Therefore, the successful transmission probability for the node $n_k$ can be expressed as follows:
\begin{equation}	
P_{ST}^{abs,k}=(P_{abs}^k+(1-P_{abs}^{tot})P_s/m_{abs})(1-P_e).
\end{equation}

Since the CSMA transmission initiated in the unreserved period in the ABS period adopts the CSMA mechanism, $P_s$ can be directly calculated using the model established in \cite{Paperbianqi}.






The average duration of ABS period is represented by $T_{abs}^{tot}$, where the average transmission duration initiated through the absolute reservation mechanism $T_{abs}^{rsv}$ and the CSMA mechanism $T_{abs}^{cs}$. Then the following equations hold: 

\begin{align}
	\left\{
	\begin{aligned}
	&T_{abs}^{rsv}=(1-P_e)T_{suc}^{rsv}+P_eT_{fai}^{rsv} \\
	&T_{abs}^{cs}=(1-P_e)P_sT_{suc}^{cs}+(1-(1-P_e)P_s)T_{fai}^{cs}+u_{bo}\sigma\\
	&T_{abs}^{tot}=P_{abs}^{tot}T_{abs}^{rsv}+(1-P_{abs}^{tot})T_{abs}^{cs}
	\end{aligned}
	\right.
,
\end{align}
where $u_{bo}$ represents the average number of backoff slots per CSMA transmission initiation.

\subsection{Calculation of Throughput and Delay}\label{CacThrouDelay}
Let $P_{ST}^{k}$ be the probability that the node $n_k$ successfully transmits in a given period. For the node using the absolute reservation, relative reservation and CSMA, it equals to $P_{ST}^{abs,k}$, $P_{ST}^{rel,k}$ and $P_{ST}^{cs,k}$ respectively.

The average transmission duration of various transmission opportunities is
\begin{equation}	
T_{all}^{tot}=P_{ran}T_{ran}^{tot}+(1-P_{ran}T_{abs}^{tot}).
\end{equation}

For a given node $n_k$, the throughput is
\begin{equation}
H^k=P_{ST}^kE[P]/T_{all}^{tot},
\end{equation}
where $E[P]$ represents the average size of payload.

The average contention delay is
\begin{equation}
T_{d}^k=T_{all}^{tot}/P_{ST}^k .
\end{equation}

\section{Simulation Result}\label{SIMULATION RESULT}
In this section, we will analyze the simulation result in three parts. Section \ref{Simulation parameter description} shows the configured parameters. Section \ref{Model fit verification} validates the theoretical analysis with the simulation. Section \ref{PERFORMANCE EVALUATION} assesses the proposed scheme against the alternative reservation mechanisms.
\subsection{Simulation Parameter Description}\label{Simulation parameter description}
The simulation platform utilized in this study is based on a simulation platform \cite{simu1} using NS-3, where the proposed DSGARM reservation mechanism is implemented. Unless stated otherwise, the parameters are set according to the specifications provided in Table \ref{tab1}. 

The Real-time Traffic Weighted Average Delay Utility (WADU) is represented by $U_w$, and its calculation formula is defined as follows:
\begin{equation}
U_w\triangleq\sum_{i=0}^{k_{low}-1}{\alpha_iU_i/\sum_{i=0}^{k_{low}-1}{\alpha_i}}.
\end{equation}

The utility weighted throughput (UWT) is defined as $H_w$, which weights the throughput through the weight of each traffic and $U_w$, which is obtained by
\begin{equation}
H_w\triangleq\sum_{i=0}^{k_{tot}-1}{\alpha_iU_iH_i}.
\end{equation}
where the $H_i$ represents the throughput of traffic $E_i$.

\begin{table}
\centering
\caption{Simulation Setup}
\begin{tabular}{cc|cc} 
\toprule 
Parameter & Value & Parameter & Value \\ 
\midrule 
Channel capacity & 21Mbps & $\alpha_0$ & 1.48 \\
Package size & 1500Bytes &$\alpha_1$ & 1.16 \\
$\varepsilon_{abs}$ & 0.42 &$\alpha_2$ & 1 \\
$\varepsilon_{res}$ & 0.514 &$\alpha_3$ &0.84 \\
$L_0^{min}$ & 34ms &$\alpha_4$ &0.72 \\
$L_0^{max}$ & 68ms &$\alpha_5$ &0.58 \\
$L_1^{min}$ & 82ms &$\lambda$ &0.85 \\
$L_1^{max}$ & 134ms &$P_{RSV}^{sui}$ & 40 \\
\bottomrule 
\end{tabular}
\label{tab1}
\end{table}

\subsection{Validation of Theoretical Analysis}\label{Model fit verification}

This section focuses on validating the precision of the analytical model. To achieve this, five groups of simulations are conducted.

Simulation scenario \uppercase\expandafter{\romannumeral1} is configured with different reservation window sizes and a total of $20$ nodes, and the proportion of ABS nodes, REL nodes, and CSMA nodes is $0:1:1$. 

In simulation scenario \uppercase\expandafter{\romannumeral2}, the proportion of the three types of nodes is $3:3:4$, and the average reservation window sizes for the absolute reservation mechanism and the relative reservation mechanism are $40$ and $80$ respectively. In this simulation, the number of nodes change from $10$ to $70$.

Simulation scenario \uppercase\expandafter{\romannumeral3} is configured with varying channel error probability, the proportion of the three types of nodes is $2:3:5$, and the reservation window size for the absolute reservation mechanism and the relative reservation mechanism are $30$ and $60$, respectively.

Simulation scenario \uppercase\expandafter{\romannumeral4} comprises $50$ nodes, where each node randomly configures its traffic type and reservation parameters. Four groups of parameters configuration are conducted, and the performance of each node in each is recorded.


To verify the probability density function of delay, simulation scenario \uppercase\expandafter{\romannumeral5} is established for three types of nodes, which are CSMA nodes, ABS nodes and REL nodes. Each simulation contains one type of nodes with the total number of $40$.

\begin{figure*}
\subfloat[Scenario \uppercase\expandafter{\romannumeral1} Reservation Ratio]{\includegraphics[width=0.33\textwidth]{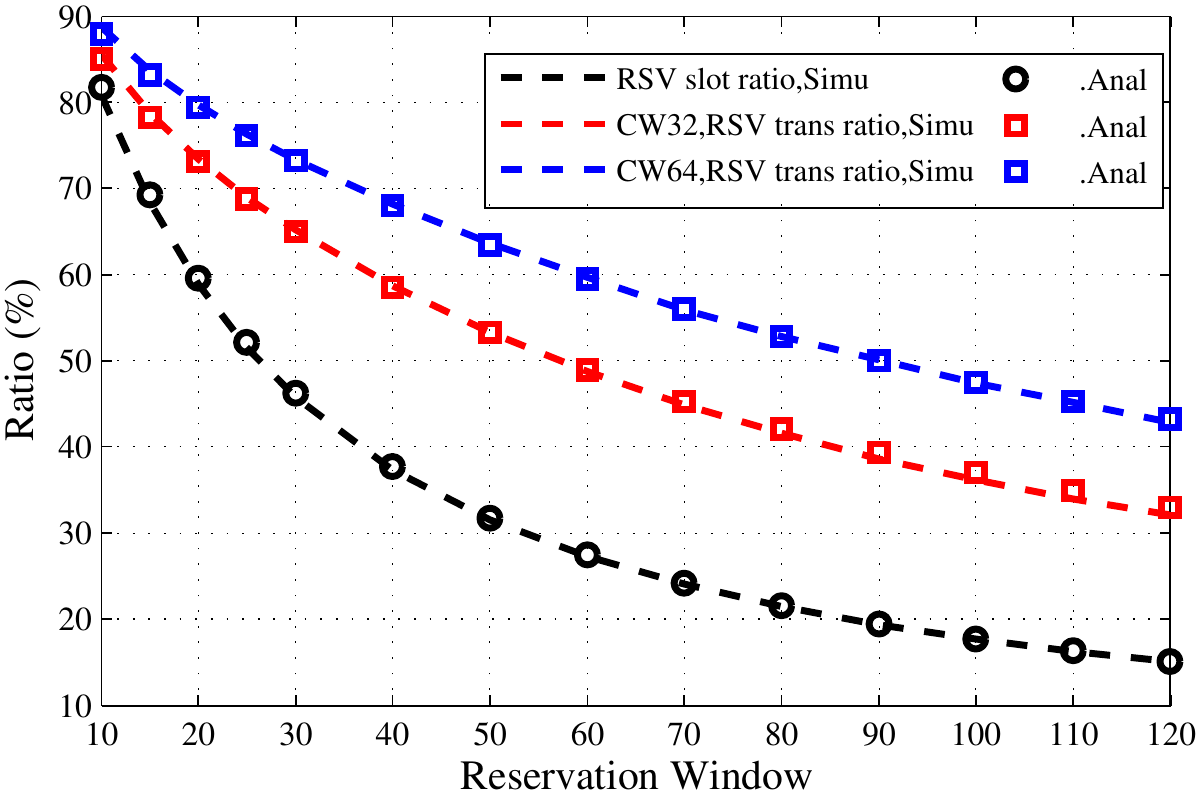}\label{figAnalysis: sub_figure2}}
\subfloat[Scenario \uppercase\expandafter{\romannumeral2} Throughput]{\includegraphics[width=0.33\textwidth]{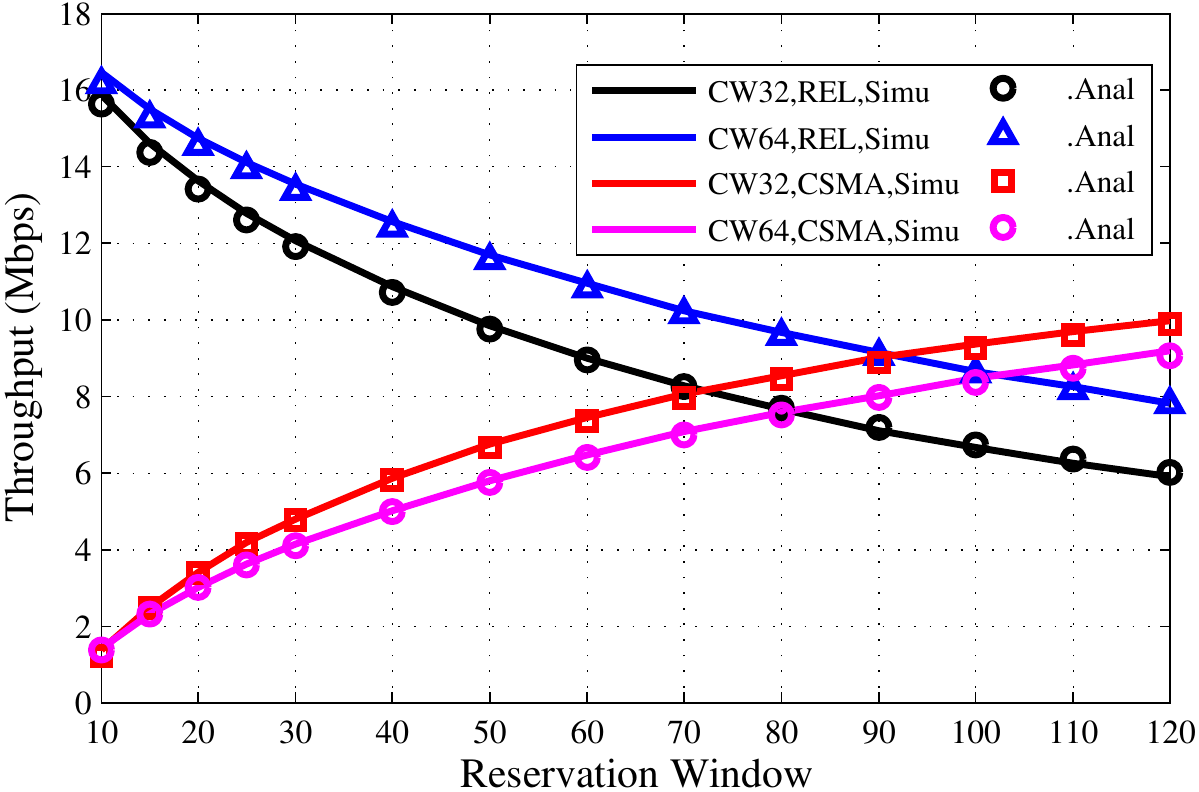}\label{figAnalysis: sub_figure3}}
\subfloat[Scenario \uppercase\expandafter{\romannumeral2} Contention Delay]{\includegraphics[width=0.33\textwidth]{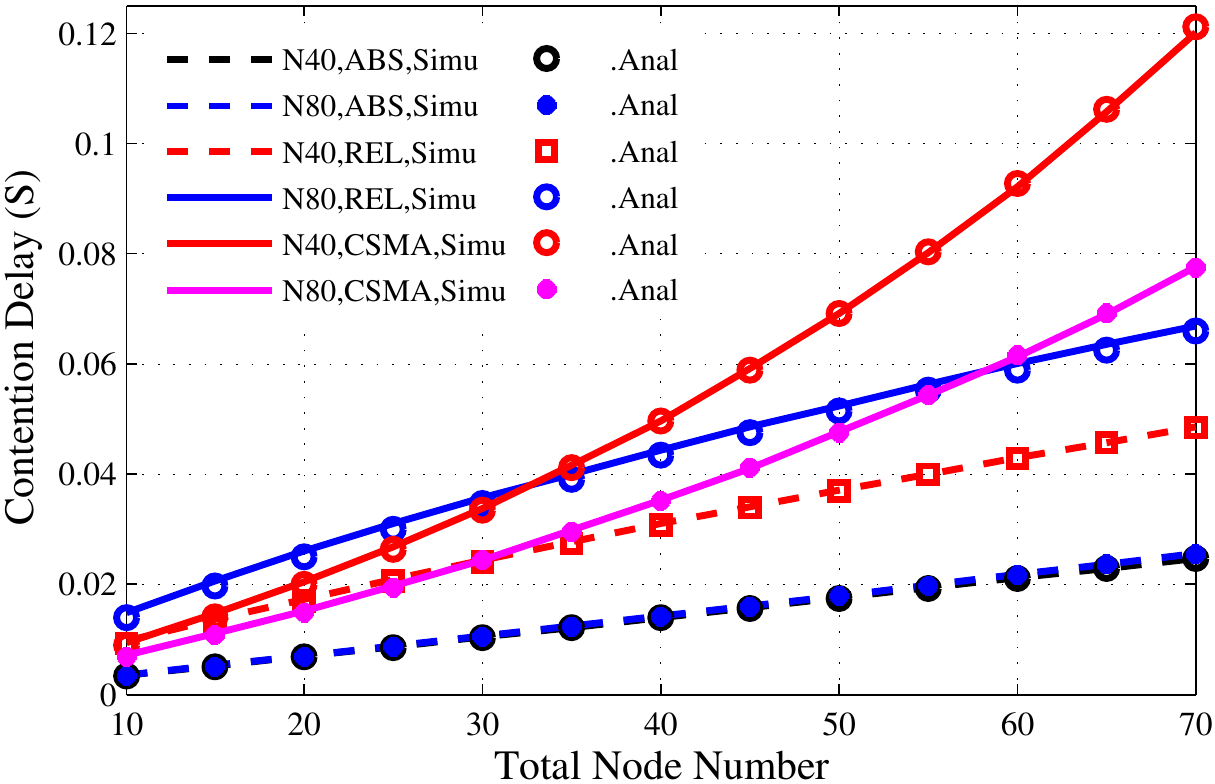}\label{figAnalysis: sub_figure5}}\\
\subfloat[Scenario \uppercase\expandafter{\romannumeral3} Reservation Ratio]{\includegraphics[width=0.33\textwidth]{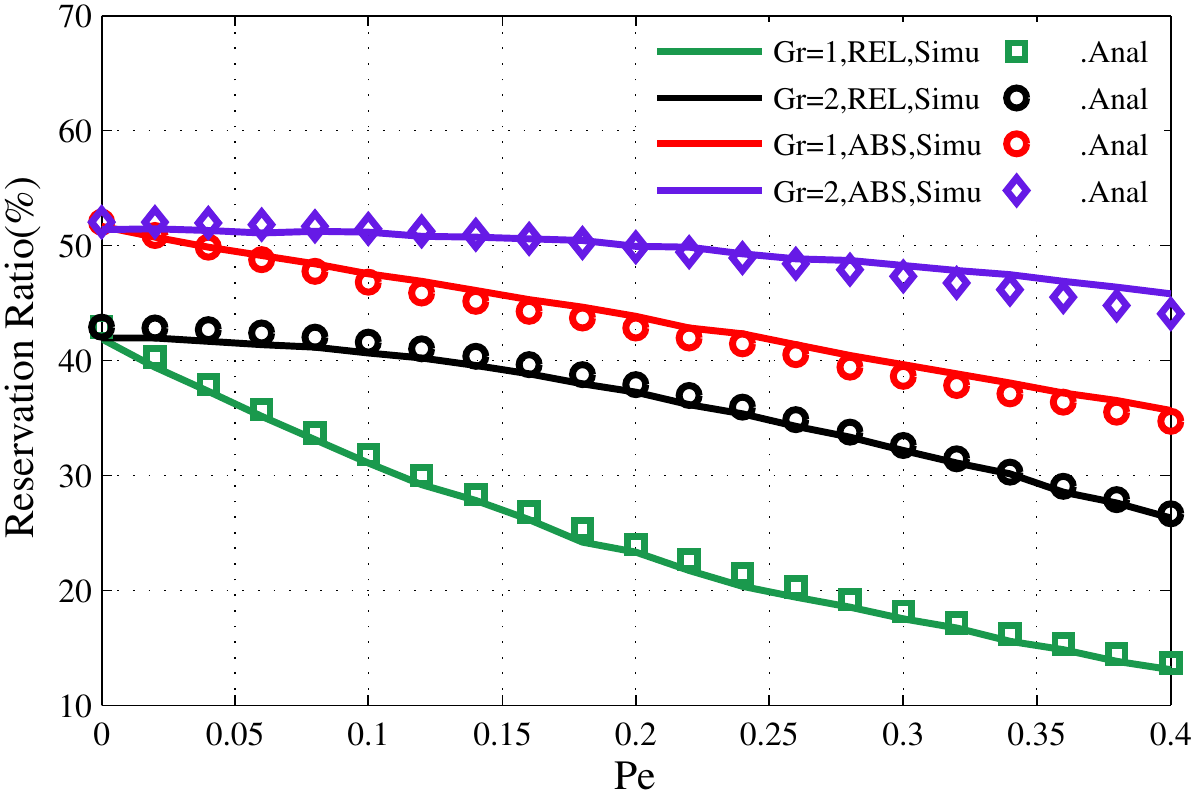}\label{figAnalysis: sub_figure9}}
\subfloat[Scenario \uppercase\expandafter{\romannumeral4} Throughput]{\includegraphics[width=0.33\textwidth]{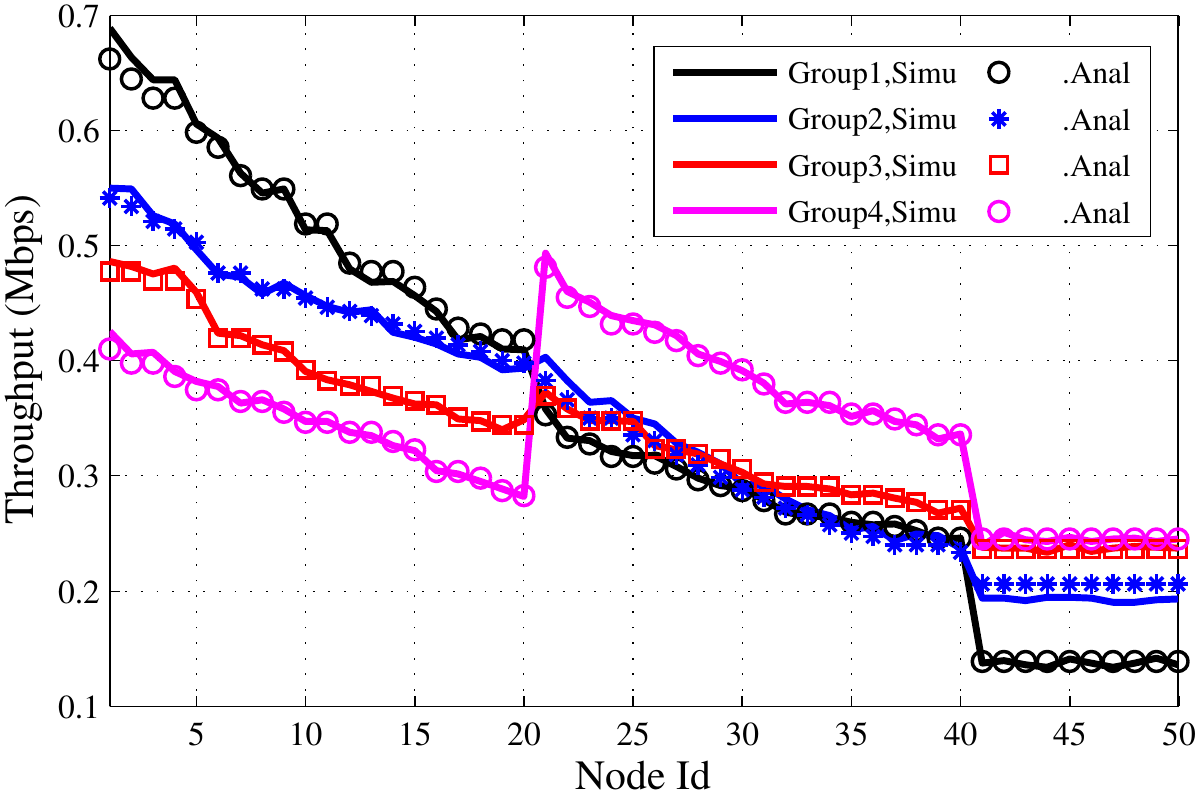}\label{figAnalysis: sub_figure7}}
\subfloat[Scenario \uppercase\expandafter{\romannumeral5} Probability Density Function]{\includegraphics[width=0.33\textwidth]{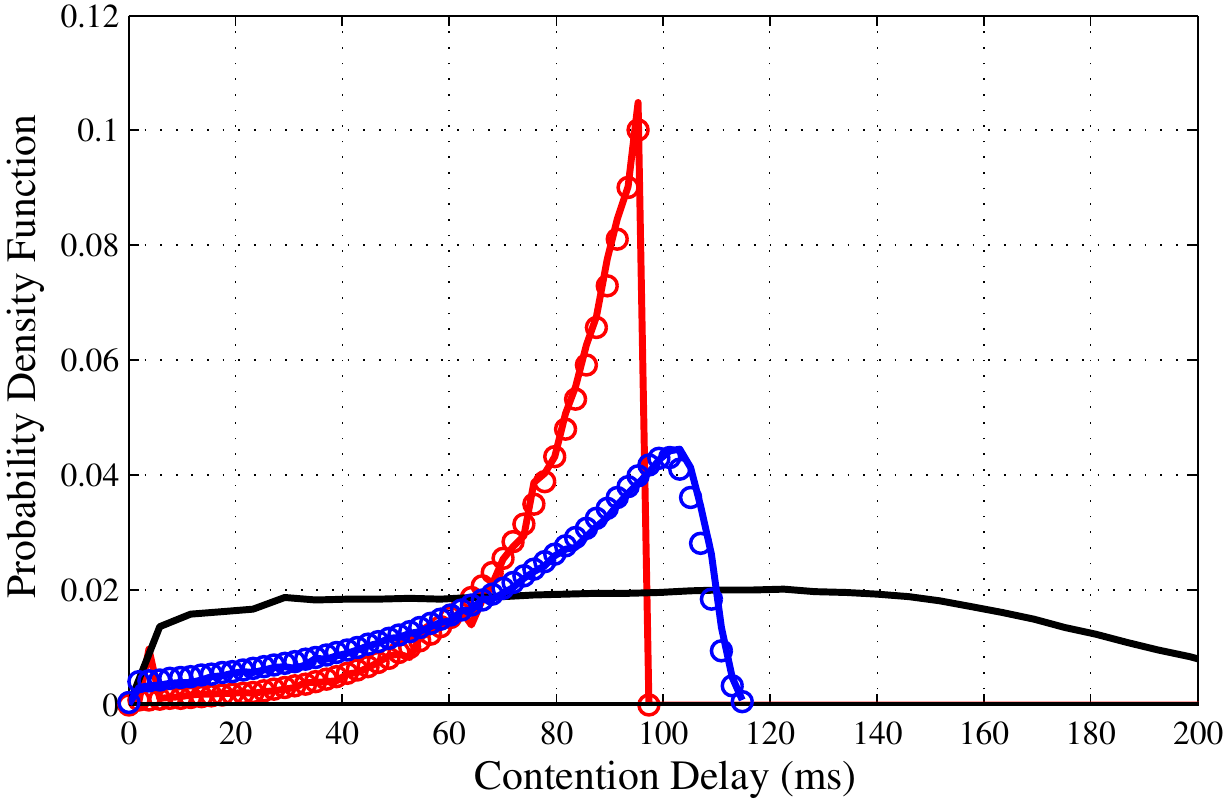}\label{figRelPdf}}

\caption{Analysis and simulation results}
\label{fig_simu_res}
\end{figure*}


The simulation performance of throughput, delay and reservation ratio in the above scenarios is consistent with the results of the analytical model.
In Fig. \ref{figAnalysis: sub_figure2}, 'RSV slot ratio' represents the proportion of slots occupied by relative reservations of all slots including the idle and CSMA slots. 'RSV trans ratio' represents the proportion of transmission occupied by relative reservations.
Fig. \ref{figRelPdf} exhibits the probability density curves of the contention delay observed with the three mechanisms. When employing the relative and absolute reservation mechanisms, the distribution of contention delay is primarily concentrated around the average delay, indicating a diminished probability of encountering extremely short or long transmission delay. The findings indicate that the reservation mechanism offers a reduced average delay and enhanced fairness in comparison to the CSMA mechanism.

\subsection{Performance Evaluation}\label{PERFORMANCE EVALUATION}
In this section, we conduct simulations to compare with the 802.11ax protocol and two reservation mechanisms, which is SCRP mechanism \cite{juedui2} and DHM-MRR mechanism \cite{juedui4}. The proposed DSGARM protocol is configured with three reservation types: pure absolute reservation, pure relative reservation, and hybrid reservation configurations. The SCRP reservation mechanism is a flexible approach that ensures real-time traffic support. On the other hand, the DHM-MRR protocol is a differentiated reservation protocol that determines the reservation period based on delay requirements. The 802.11ax protocol \cite{simu4} utilizes the EDCA mechanism.

To compare the three reservation mechanisms fairly, we maintain a consistent proportion of reservation resources for real-time traffic in the simulation of the three mechanisms. 

\subsubsection{Performance Comparison of Varying Node Number}\label{Performance comparison of variable node number}
\ 
\newline 
\indent In the simulation of this section, the traffic ratios from $E_0$ to $E_5$ are configured to be $10\%$, $20\%$, $10\%$, $10\%$, $20\%$ and $30\%$ respectively, where $E_0$ and $E_1$ belong to the real-time traffic. The simulation results with varying numbers of nodes are shown in Fig. \ref{fig Performance comparison of variable node number}. Due to the large delay of 802.11ax when there are more nodes, logarithmic processing is applied to the results of delay.

\begin{figure*}[htb]

\captionsetup{position=bottom}
\subfloat[Ordinary Traffic Throughput]{\includegraphics[width=0.33\textwidth]{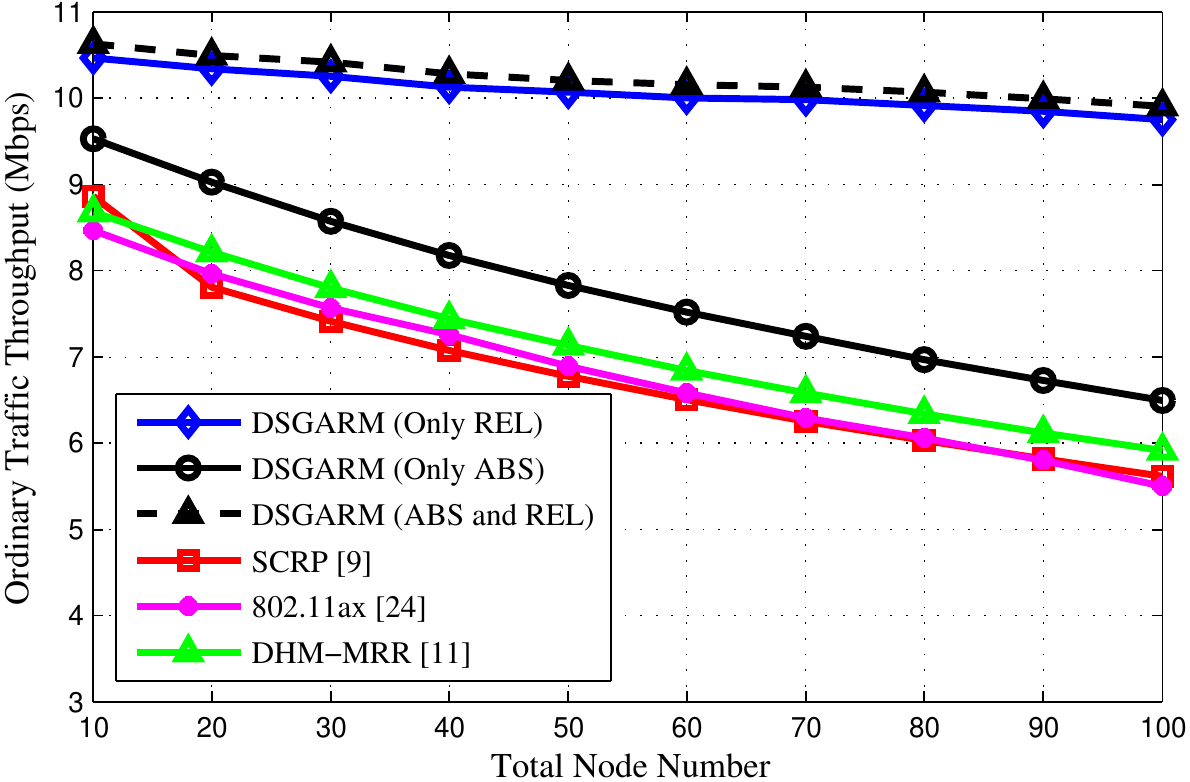}\label{figChanNode: sub_figure2}}
\subfloat[Utility Weighted Throughput]{\includegraphics[width=0.33\textwidth]{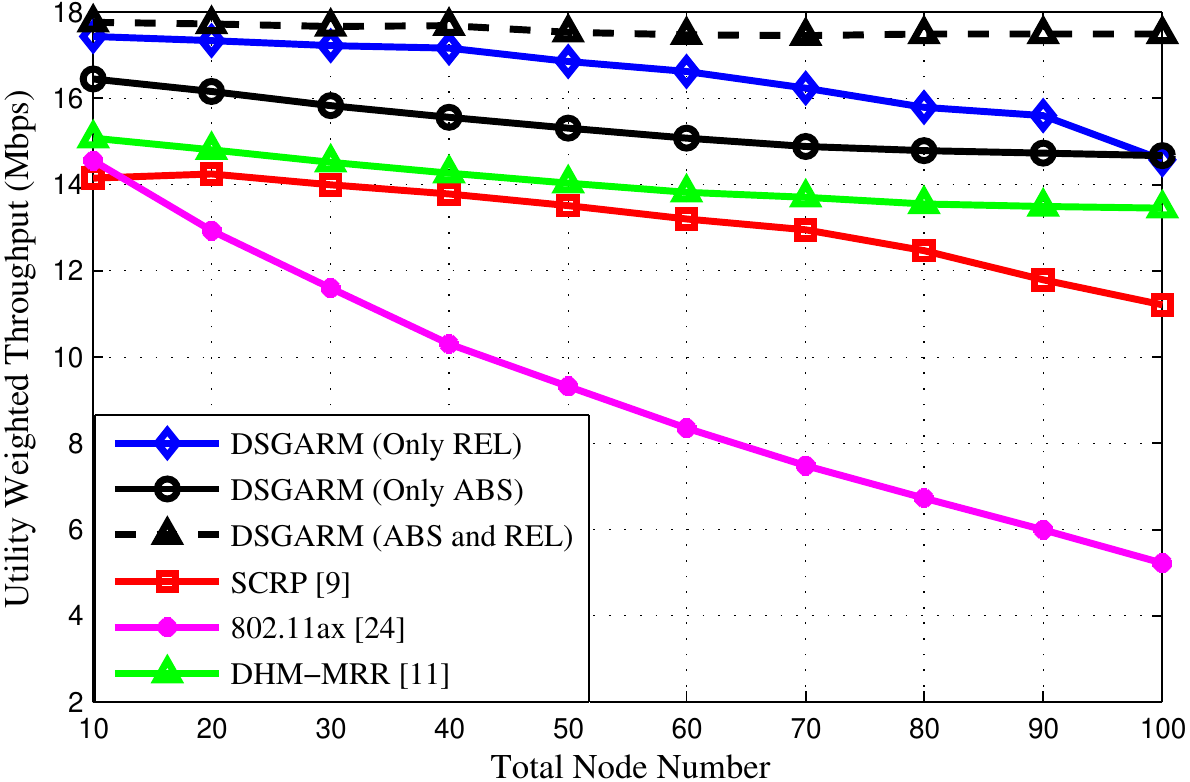}\label{figChanNode: sub_figure3}}
\subfloat[Weighted Average Delay Utility]{\includegraphics[width=0.33\textwidth]{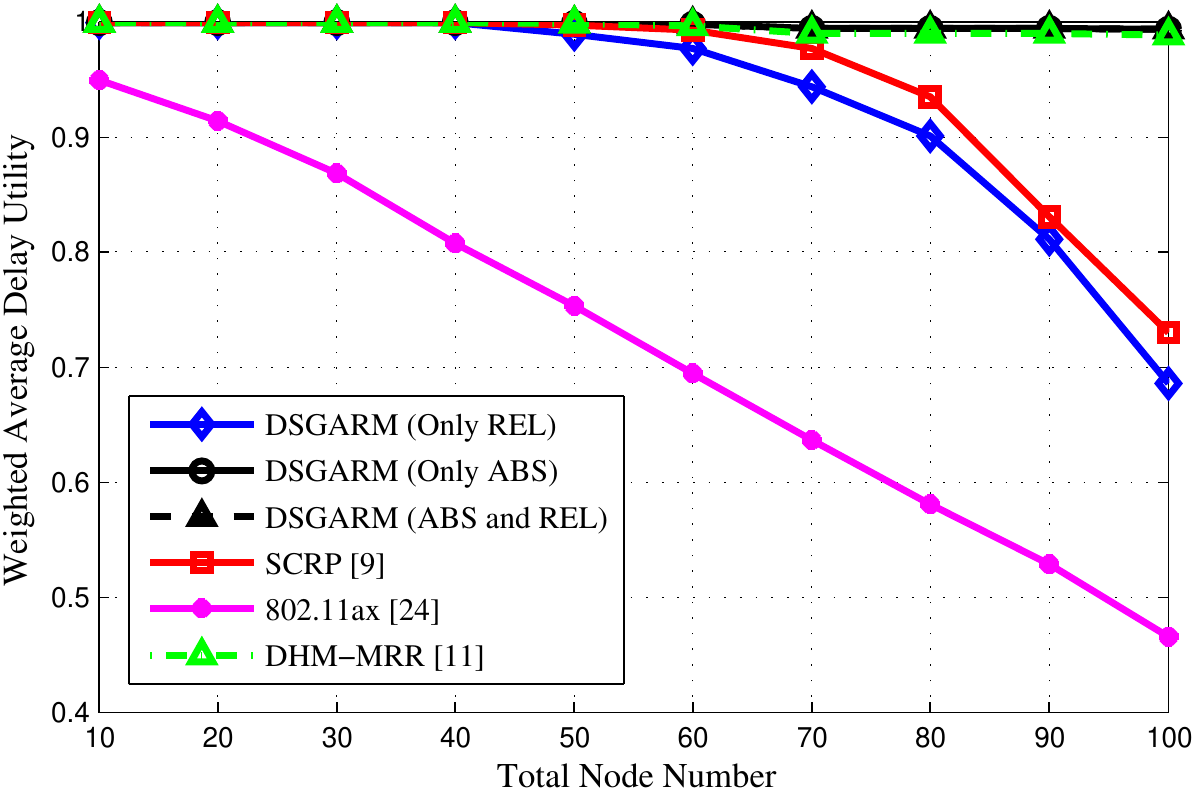}\label{figChanNode: sub_figure4}}\\
\subfloat[Real-time traffic delay]{\includegraphics[width=0.33\textwidth]{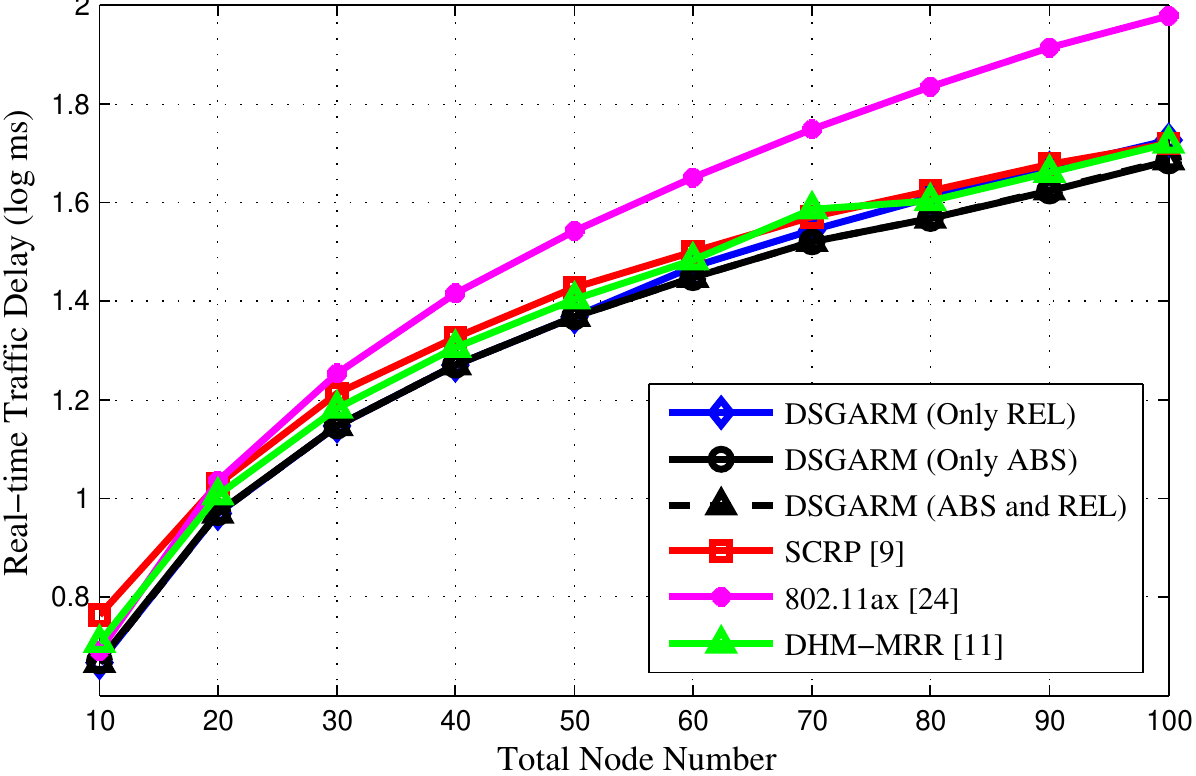}\label{figChanNode: sub_figure1}}
\subfloat[Weighted Average Delay Utility]{\includegraphics[width=0.33\textwidth]{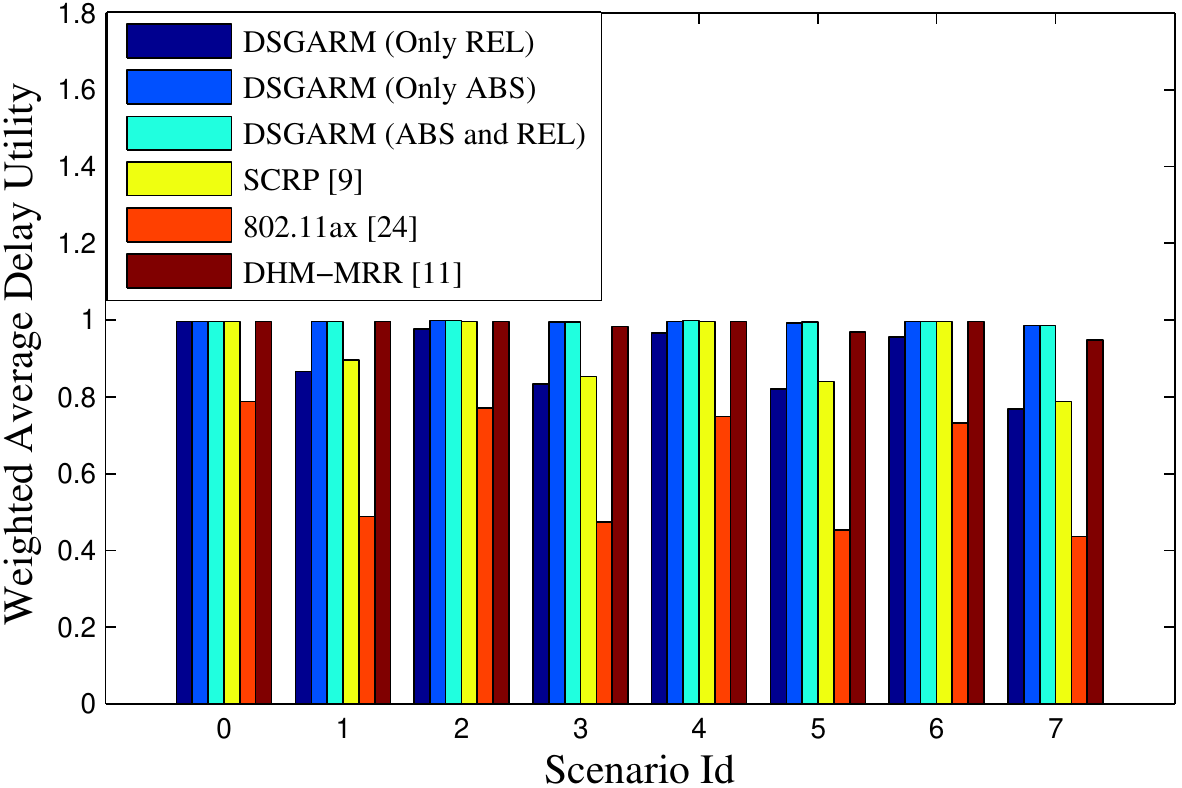}\label{figChanTra: sub_figure1}}
\subfloat[Utility Weighted Throughput]{\includegraphics[width=0.33\textwidth]{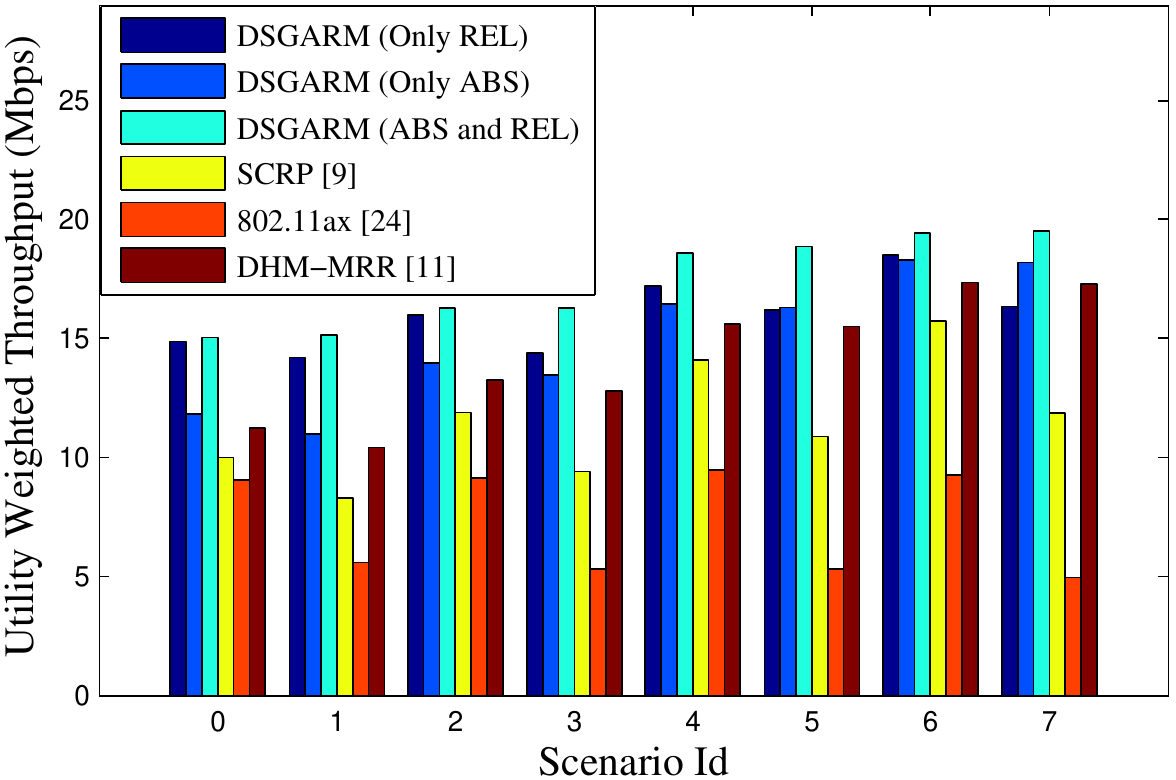}\label{figChanTra: sub_figure2}}\\
\caption{Performance comparison}
\label{fig Performance comparison of variable node number}
\end{figure*}


Extensive simulations suggest that the DSGARM protocol, employing the hybrid reservation mechanism, outperforms the pure absolute or relative reservation methods. Regarding the throughput of ordinary traffic, the hybrid protocol exhibits superior performance compared to the pure absolute reservation approach, owing to the enhancements in fragmentation and relative reservation. Additionally, the pure relative reservation enhances the transmission performance of ordinary traffic and also decreases the delay of real-time traffic compared to the 802.11ax protocol. However, due to the inability to ensure strict delay guarantees, the WADU gradually declines with an increase in the number of nodes. Consequently, the exclusive usage of either relative or absolute reservation cannot guarantee optimal performance for all types of traffic.


Compared to other reservation schemes, the proposed DSGARM protocol demonstrates superior performance by offering accurate QoS guarantee through the adaptive algorithm for the absolute reservation, effectively meeting the real-time traffic demands even under heavy loads, while employing relative reservation for ordinary traffic. The DHM-MRR can also guarantee the real-time traffic requirements. However, it exhibits inadequate performance for ordinary traffic due to the fragmentation issue associated with absolute reservation. While the CSRP protocol shows notable throughput improvements over the 802.11ax, it falls short in accurately guaranteeing the demands of real-time traffic and fails to ensure optimal performance for ordinary traffic, as the fragmentation problem reduces the throughput of ordinary traffic. Consequently, as the number of nodes increases, the overall throughput of this protocol is relatively lower.

\subsubsection{Performance Comparison of Varying Traffic}\label{Change scenario service type mechanism performance comparison}
\ 
\begin{table}[htbp]
    \centering
    \caption{Proportional configuration of traffic}
    \label{TraficPro}
    \begin{tabular}{ccccc}
        \toprule
        \textbf{Scenario Id} & \textbf{0-1} & \textbf{2-3} & \textbf{4-5}& \textbf{6-7} \\
        \midrule
        $E_0$ & 5\% & 10\% & 15\% & 20$\%$\\
        $E_1$ & 10\% & 15\% & 20\% & 25$\%$\\
        $E_2$ & 15\% & 11\% & 10\%  & 8$\%$\\
        $E_3$ & 15\% & 11\% & 10\%  & 8$\%$\\
        $E_4$ & 15\% & 13\% & 20\% & 8$\%$ \\
        $E_5$ & 40\% & 40\% & 25\%  & 8$\%$\\
        \bottomrule
    \end{tabular}
\end{table}

Table \ref{TraficPro} provides the proportion of each traffic type in each simulation scenario. Notably, scenarios with even IDs involve 60 configured nodes, while scenarios with odd IDs comprise 100 configured nodes. The simulation results are depicted in Fig. \ref{figChanTra: sub_figure1} and Fig. \ref{figChanTra: sub_figure2} respectively.

The simulation results suggest that, when the proportion of real-time traffic in the network is relatively low, both the DSGARM and DHM-MRR protocols effectively preserve the WADU of $1$. Furthermore, these two protocols outperform the comparison schemes in maintaining a utility of $1$ for more scenarios. In scenarios with the proportion of real-time traffic is high, the DSGARM protocol achieves a comparatively higher WADU by employing more efficient resource allocation strategies. The hybrid reservation mechanism also ensures optimal performance for ordinary traffic, leading to superior performance when compared to pure absolute reservation or relative reservation approaches. Overall, the DSGARM protocol exhibits the highest UWT across all considered scenarios.

\section{Implementation Issues}\label{IMPLEMENTATION ISSUES}
\begin{figure}
\subfloat[Impact of NCSC]{\includegraphics[width=0.25\textwidth]{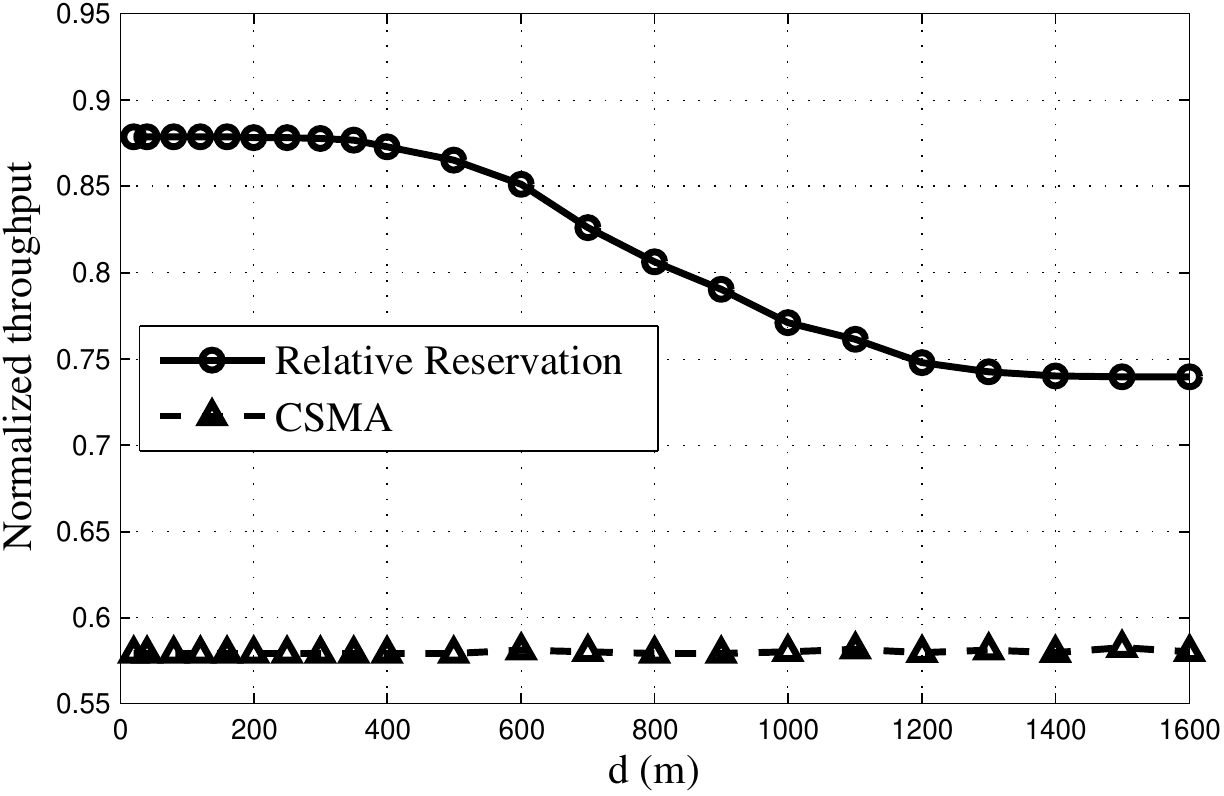}\label{figCHANGEd}}
\subfloat[Explanation of Convergence Time]{\includegraphics[width=0.25\textwidth]{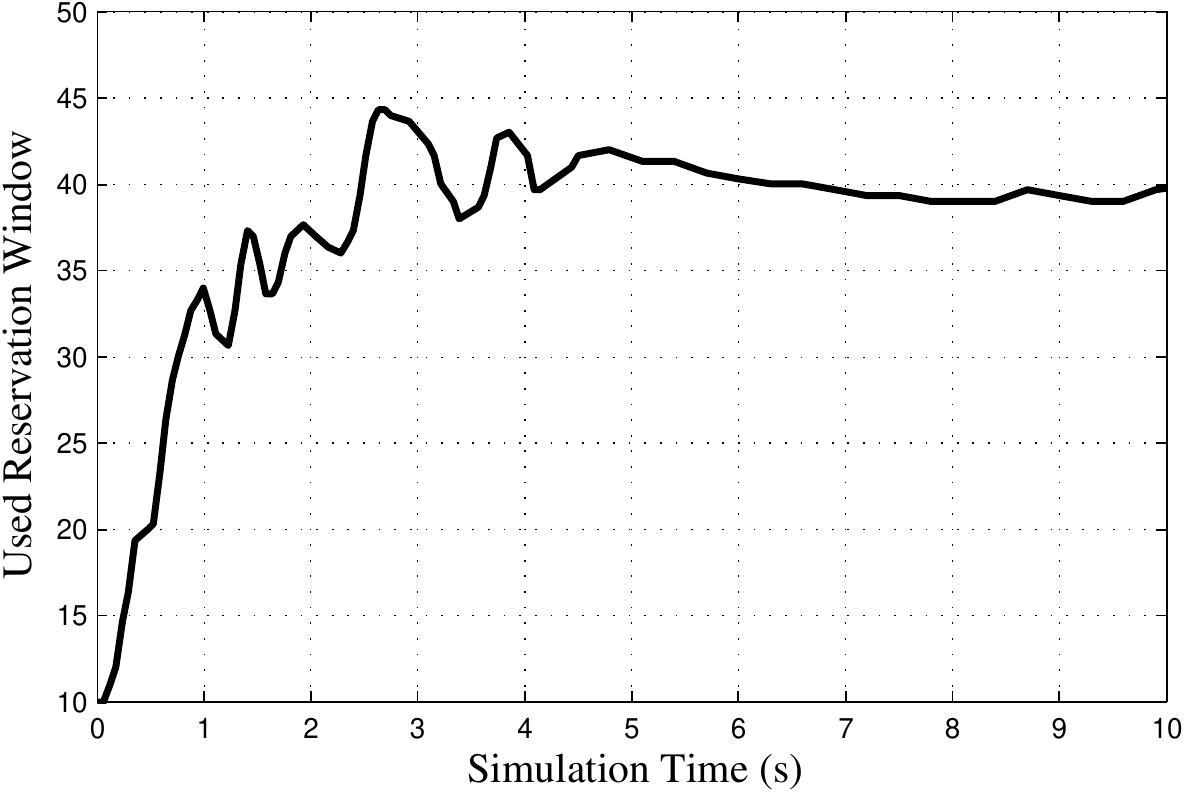}\label{figShouLian}}
\caption{Simulation result of implementation}
\label{fig:simulation result of implementation}
\end{figure}
\subsection{Traffic Prediction}\label{Traffic Torecasting}
The effectiveness of the channel reservation mechanism relies on precise traffic prediction or the traffic type characterized by regular arrivals. There is a potential risk of the lack of corresponding traffic upon the arrival of the reservation time. Researches \cite{Pro1}, \cite{Pro2} and \cite{wangluming} have already proposed accurate traffic prediction techniques. Additionally, the proposed soft reservation mechanism mitigates the overhead caused by this phenomenon. The aforementioned works can effectively address this issue. Therefore, the impact of burst traffic on reservations falls outside the scope of this paper.

\subsection{Mitigating the Impact of Algorithmic Computation}\label{Mitigating the Impact of Algorithmic Computation}

In order to mitigate the potential impact of algorithmic computation on real-time traffic transmission, the algorithm proposed in this paper adopts a periodic deferral computation during idle periods. Moreover, the distribution of reservation weight information across nodes and the sharing of computation results through reservation information further minimize the influence of the algorithmic execution process on the transmission process.


The proposed algorithm demonstrates rapid convergence within a relatively short time. In the simulation scenario involving $60$ nodes, as discussed in Section \ref{Performance comparison of variable node number}, the convergence of the reservation window size for observed node decisions is shown as Fig. \ref{figShouLian}. Specifically, stable parameters are achieved within $5$ seconds of the network being operational.



\subsection{Problem of Neighbor Channel Sensing Capability}\label{Problem of neighbor channel sensing capability}

\ 
\newline 
\indent Although this paper primarily concentrates on the densely deployed scenario, where nodes typically maintain consistent channel perceptions. There are instances where nodes may experience inconsistent channel perceptions with their neighbors due to long distances. This section investigates the performance of relative reservation in such cases. The channel model is been configured according to \cite{quqiaoXindao}. Specifically, a fixed observation node $A$ is employed, and $m$ nodes are randomly generated within the $[0,2d]$ around node $A$, following the uniform distribution.


The simulation results are shown in Fig. \ref{figCHANGEd}. In scenarios where nodes are situated nearby, the performance of the relative reservation mechanism remains well, as the channel perceptions between nodes are consistent. As the distance between nodes increases, variations in channel perceptions lead to reduced accuracy in relative reservation and a subsequent decrease in throughput performance. Nevertheless, even in cases where nodes are significantly distant from each other, the relative reservation mechanism still outperforms the CSMA protocol. It shows that for node pairs with substantial distance, although the relative reservation may not accurately prevent collisions, the erroneous relative reservation information can be viewed as an increase in the backoff window size, resulting in relatively minor additional overhead and even potentially reducing collisions. Conversely, for the close node pairs, the relative reservation effectively mitigates collisions. Therefore, the incorporation of relative reservation remains effective in scenarios where nodes are in close proximity to each other.

\subsection{The Accurate Submission of Reservation Information}

Owing to the possibility of a substantial complicated range, the reservation information forwarded via ACK may not be adequate to guarantee that all hidden terminals receive the reservation information. In \cite{yuanyunjie}, a method is proposed to address hidden terminals by employing multiple forwarding and relaying of reservation information, thereby achieving accurate submission of reservation information in complex interference scenarios. This method can be integrated with the protocol proposed in this paper.

\section{Conclusions}\label{Conclusions}
This paper introduces a hybrid channel reservation mechanism that provides differentiated traffic guarantees to ensure QoS for various traffic types. Through simulation results, the theoretical analysis is validated, highlighting the superiority of the proposed mechanism compared to existing reservation mechanisms.


Recently, the research group of the author has proposed the idea of particle-based access \cite{lizihua} and the network wave theory \cite{wangluobo}, both of which are important methods for ensuring traffic transmission quality and orderly access. The channel reservation mechanism is an important method for supporting and implementing the above ideas and theories. Therefore, the reservation protocol proposed in this paper will be combined with them in subsequent extensions.

\section{Acknowledgement}
This work was supported in part by the National Natural Science Foundations of CHINA (Grant No. 61871322, No. 61771390, and No. 61771392).

\bibliographystyle{IEEEtran}
\bibliography{New_IEEEtran_how-to}

\end{document}